\documentclass[lettersize, journal]{IEEEtran}

\usepackage[english]{babel}
\usepackage[utf8]{inputenc}
\usepackage{microtype}
\usepackage{amsmath,amsfonts}
\usepackage{graphicx}
\usepackage{nicematrix} 
\usepackage{yaacro}
\usepackage{amssymb}    
\usepackage{booktabs}
\usepackage{siunitx}
\usepackage{nicefrac}
\usepackage{amsthm}		
\usepackage{cite}
\usepackage{subfig}
\usepackage{color}
\usepackage{soul}

\usepackage[colorinlistoftodos]{todonotes}
\presetkeys{todonotes}{inline,backgroundcolor=orange}{}

\newtheorem{theorem}{Theorem}
\newtheorem{corollary}{Corollary}
\newtheorem{lemma}{Lemma}
\newtheorem{remark}{Remark}

\newtheorem{proposition}{Proposition}

\newtheorem{assumption}{Assumption}

\begin{acgroupdef}[list=acronimos]
    \acdef{CNPq}{Conselho Nacional de Desenvolvimento Científico e Tecnológico}
    \acdef{CAPES}{Coordenação de Aperfeiçoamento de Pessoal de Nível Superior}
    \acdef{FAPEMIG}{Fundação de Amparo à Pesquisa do Estado de Minas Gerais}
    \acdef{FAPESP}{Fundação de Amparo à Pesquisa do Estado de São Paulo}
    \acdef{AHS}{Automated Highway Systems}
    \acdef{CC}{Cruise Control}
    \acdef{ACC}{Adaptive Cruise Control}
    \acdef{CACC}{Cooperative Adaptive Cruise Control}
    \acdef{CD}{Constant Distance}
    \acdef{V2V}{Vehicle-to-Vehicle}
    \acdef{V2I}{Vehicle-to-Infrastructure}
    \acdef{V2X}{Vehicle-to-Everything}
    \acdef{IFT}{Information Flow Topology}
    \acdef{PF}{Predecessor Following}
    \acdef{BD}{Bidirectional}
    \acdef{TPF}{Two-Predecessor Following}
    \acdef{TPFL}{Two-Predecessor Following Leader}
    \acdef{BDL}{Bi-directional Leader}
    \acdef{PFL}{Predecessor Following Leader}
    \acdef{rPFL}{$r$-Predecessor Following Leader}
    \acdef{CD}{Constant Distance}
    \acdef{BIBO}{Bounded-Input Bounded-Output}
    \acdef{GVM}{Gross Vehicle Mass}
    \acdef{RG}{Road Grade}
    \acdef{SMC}{Sliding Mode Control}
    \acdef{MPC}{Model Predictive Control}
\end{acgroupdef}

\begin{document}


\title{Improving the Design of Linear Controllers for Homogeneous Platooning under Disturbances}

\author{Emerson A. da Silva$^{1}$, Leonardo A. Mozelli$^{2}$, Armando A. Neto$^{2}$, and Fernando O. Souza$^{2}$ 
\thanks{This study was partially supported by \ace{CAPES} - codes 001 and 88887.136349/2017-00, \ace{CNPq} - grants 465755/2014-3, 429819/2018-8 and 	312034/2020-2, \ace{FAPESP} - grant 2014/50851-0, and \ace{FAPEMIG} - grant APQ-00543-17.}
\thanks{E. A. da Silva is with the Graduate Prog. in Electrical Engineering, Univ. Fed. de Minas Gerais, Belo Horizonte, Brazil. {\tt\small eas2011@ufmg.br}}
\thanks{A. A. Neto, F. O. Souza and L. A. Mozelli are with the Dep. of Electronics Engineering, Univ. Fed. de Minas Gerais, Belo Horizonte, Brazil. {\tt\small \{aaneto,fosouza,mozelli\}@cpdee.ufmg.br}.}
%
}

\markboth{IEEE TRANSACTIONS ON INTELLIGENT TRANSPORTATION SYSTEMS, Submitted in 2023}%
{How to Use the IEEEtran \LaTeX \ Templates}

\maketitle



\begin{abstract}
    

    This paper addresses the problem of longitudinal platooning control of homogeneous vehicles subject to external disturbances, such as wind gusts, road slopes, and parametric uncertainties. Our control objective is to maintain the relative distance of the cars regarding their nearby teammates in a decentralized manner. Therefore, we proposed a novel control law to compute the acceleration commands of each vehicle that includes the integral of the spacing error, which endows the controller with the capability to mitigate external disturbances in steady-state conditions. We adopt a constant distance spacing policy and employ generalized look-ahead and bidirectional network topologies. We provide formal conditions for the controller synthesis that ensure the internal stability of the platoon under the proposed control law in the presence of constant and bounded disturbances affecting multiple vehicles. Experiments considering nonlinear vehicle models in the high-fidelity CARLA simulator environment under different disturbances, parametric uncertainties, and several network topologies demonstrate the effectiveness of our approach.
    
\end{abstract}


\section{Introduction}\label{sec:Introduction}
	
\IEEEPARstart{T}{he} growing demand for passenger and freight transportation has led to an increase in road usage and a decline in traffic flow efficiency. Vehicle platooning constitutes a possible solution to address these problems. A platoon is a group of connected and autonomous vehicles that move cooperatively with a short inter-vehicle distance and synchronized speed. They allow better usage of road space, reduced fuel consumption and environmental emissions, and enhanced driving safety.

These benefits rely heavily on the platoon's ability to keep a desired formation during travel. However, in practical driving conditions, internal and external disturbances can cause significant spacing errors that compromise the platoon formation and degrade its performance.

Precise models of engine, clutch, gearbox, wheel, tire, and braking systems are usually nonlinear, and their parameters are affected by vehicle aging and changes in the environment/roads \cite{Gao_2016_Article}. The latter can be described as external disturbances acting on each vehicle in the platoon, such as road slopes and wind gusts, as illustrated in Fig.~\ref{fig:problemdef}.

\begin{figure}[ht]
	\centering
	\includegraphics[width=\linewidth]{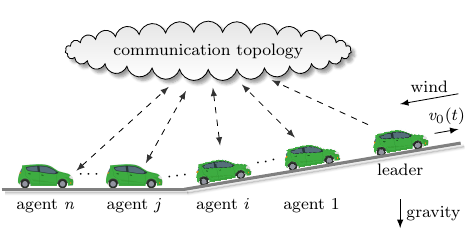}
	\caption{Homogeneous vehicular platoon subject to external disturbances.} 
	\label{fig:problemdef}
\end{figure}

In this work, we deal with the longitudinal control of homogeneous vehicular platoons under a \ac{CD} spacing policy, fixed network topology, free of communication delays. We propose an effective distributed control law that achieves null steady-state spacing error for general directed and undirected network topologies in a platoon subjected to road slopes, wind gusts, and parametric uncertainties modeled as bounded step disturbances in the vehicles in the platoon.


The remainder of this paper is organized as follows. Related works are presented in Sec. \ref{sec:RelatedWorks}. In Sec. \ref{sec:ProblemDefinition} we formulate the problem, present the longitudinal vehicle model, the network topology, the proposed distributed control law, and the platoon's formation error dynamics. The stability analysis of the platoon for directed and undirected topologies is given in Sec. \ref{sec:StabilityAnalysis}. Experiments in the realistic \emph{CARLA} driving environment demonstrating the effectiveness of our proposal are presented in Sec. \ref{sec:ResultsAnalysis}. Sec. \ref{sec:Conclusion} concludes the paper.

\section{Related Works}\label{sec:RelatedWorks}
	

\ac{GVM} and \ac{RG} represent the primary internal and external adversities for vehicle automation, as both exert a substantial influence on its dynamics.


Ensuring safe and dependable operation demands precise estimations of \ac{GVM} for various controllers. Its accurate knowledge enhances the performance of crucial systems such as anti-lock braking, collision avoidance, and stability control. Even in the case of small passenger cars, deviations from one trip to another can reach $50\%$, and for heavy-duty vehicles, the situation becomes even more critical. Besides safety, vehicle overload has societal impacts: increased wear and tear on roads and unfair competition among transport companies.


Terrain conditions play a pivotal role in shaping energy demand, with notable disparities in energy consumption observed between journeys across steep and flat terrain, adding up to a $20\%$ difference \cite{Boriboonsomsin_TRR_2009}. In the context of Eco-driving \ac{CC}, energy savings are attained through strategic maneuvers, including deceleration prior to descents and acceleration preceding ascents. Also, automated features, such as Hill Start Assist and transmission scheduling, are influenced by the surface gradient, having enduring impacts on both wear and operational costs.
 

To overcome these and other challenges, vehicle automation has become increasingly prevalent. Nevertheless, a recent study on commercially available \ac{ACC} unveiled potential drawbacks, including higher energy consumption and safety concerns if this technology becomes predominant within the current fleet \cite{Ciuffo_TRC_2021}. The study identified that minor speed fluctuations in the leading vehicle, caused by road geometry, could give rise to stop-and-go waves in the formation. Vehicular information and connectivity also present opportunities for addressing these issues, providing the next generation of maps that include \ac{RG} \cite{Zhang_TITS_2020}. However, limitations in the accuracy of current 3D maps and refresh rates of Global Navigation Satellite System (GNSS) sensors hinder the implementation of effective strategies to compensate for the influence of road slopes.

In the following, we outline some theoretical and practical works on vehicle platooning, focusing on the control strategies capable of handling disturbances, especially \ac{GVM} and \ac{RG}.


The work in \cite{Seiler_2004_Article} stands as a landmark that focuses on the theoretical study of the propagation of disturbances from the leader to the followers in the platoon, centered around the concept of string stability. It encompasses the \ac{PF}, \ac{PFL}, and \ac{BD} network topologies, with a CD spacing policy and a linear distributed controller for homogeneous double integrator vehicle dynamics. A more complete set of topologies is investigated in \cite{Zheng_2016_Article}. Besides many look-ahead topologies, there are investigations on bidirectional networks in which an agent exchanges information with preceding and succeeding neighboring vehicles. A third-order longitudinal model is obtained using exact feedback linearization. However, since the proposed controller primarily aims at stabilizing the formation, it lacks robustness to internal and external disturbances. Later, the authors improved the control design by considering the $H_\infty$ criteria in \cite{Zheng_2018_Article}.


Recently, \cite{Li_2022_Article} considered the same third-order model, but explicitly accounting for external disturbances, regarded as unknown signals with bounded amplitudes. However, the vehicles are assumed to be consistently traveling on flat roads. In \cite{Zhao_2022_Articlea}, Zhao et. al adopt the same vehicle model and assumptions about external disturbances. Nevertheless, both works consider another significant source of external disturbance in platooning: imperfect communication. To address the uncertainty of the \ac{PFL} topology, Li et al. \cite{Li_2022_Article} employed an adaptive controller. Meanwhile, \cite{Zhao_2022_Articlea} investigated Denial-of-Service (DoS) attacks in look-ahead topologies, modeled as temporary disruptions in the communication links, mitigated by linear distributed controller gains computed through Linear Matrix Inequalities (LMIs). The effects of time delay in the stability and steady-state conditions are examined in \cite{Souza_2020_Article} and \cite{Neto_2019_Article}, while \cite{Godinho_2022_Article} studies the merge and split phenomena. The presence of mixed traffic is addressed in \cite{Godinho_2023_Article}.


To compensate for disturbances, nonlinear control strategies have also been developed. In \cite{Kwon_2014_Article}, the focus is on the \ac{BD} topology. An adaptive \ac{SMC} strategy requiring measurements from position, velocity, and acceleration is devised. The controller is adaptable to uncertain and bounded parameters, encompassing mass, drag coefficient, and rolling friction. Additionally, adaptation extends to a lumped parameter representing the remaining uncertainties and disturbances. Recently, \cite{Zhu_TITS_2022} considered a broader set of topologies and improved the performance of the \ac{SMC} by incorporating a disturbance observer for the acceleration state. In general, \ac{SMC} is quite robust but prone to chattering phenomena and large control inputs in the presence of high model uncertainty.


Optimization-based techniques are another alternative for handling disturbances, offering the advantage of accommodating several states and input constraints in the design process. In contrast to previous approaches, a second-order model is considered in \cite{Hu_TITS_2022}, assuming instantaneous torque application and linearization around operating points. \ac{RG} and \ac{GVM} are treated as deterministic noises, whereas the air drag acts as a stochastic noise. Simulation results reveal that the proposed \ac{MPC} method effectively minimizes steady-state errors compared to traditional and robust \ac{MPC} versions. In \cite{Zhou_TITS_2022}, the model integrates terms related to actuation lag and attenuation while incorporating a distinct term to account for lumped uncertainties. A closed-loop min-max \ac{MPC} model is formulated, using a causal disturbance feedback mechanism that parameterizes the control input through a causal structure. However, the proposed centralized controller operating within the leading vehicle represents a disadvantage to previous decentralized strategies. Focusing on the \ac{PF} topology and the challenges posed by road geometry, \cite{Zhai_TVT_2018} presented a distributed \ac{MPC} based on a discretized model similar to the previous one. The study addresses the challenges of quickly obtaining online solutions with an emphasis on enhancing fuel economy.


A distinguishable paradigm involves leveraging explicit knowledge of disturbances. In \cite{Zhai_2022_Article}, information from GNSS and road maps is explored, allowing the segmentation of upcoming road into multiple sections based on road grade information. Each section has a constant slope value, irrespective of its length. Based on the premise that the dominant disturbance acting on the platoon comes from road geometry, a disturbance observer is employed to estimate the road gradient in \cite{Na_2019_Article}. Fusing this information with GNSS data, the study directly incorporates the estimated slope into the calculation of an optimal velocity profile, resulting in enhanced fuel savings.

Based on this review of recent works, it is noticeable that the topic of disturbance rejection from \ac{GVM} and \ac{RG} has deserved less attention from the Intelligent Transportation Systems (ITS) community when linear and distributed control strategies are employed, in contrast with their nonlinear and optimization-based counterparts. In the following, this gap is addressed and a solution to improve the longitudinal formation is proposed.

\section{Problem Statement}\label{sec:ProblemDefinition}

This section delineates the longitudinal control of a homogeneous platoon of vehicles. First, the longitudinal vehicle dynamic is described, followed by the network topology model. Then, a modified distributed control law is presented, with the formation error dynamics of the platoon.

\textbf{Notation}: Vectors and matrices are denoted by bold lowercase and uppercase letters, respectively. 
For $s \in \mathbb{C}$, denote $\operatorname{Re}\{s\}$ and $\operatorname{Im}\{s\}$ its real and imaginary parts. $\mathbf{I}_n$ is an $n \times n$ identity matrix. $\mathbf{1}_{n} \in \mathbb{R}^{n}$ denotes a vector of ones. $\mathbf{1}_{n}^{i} \in \mathbb{R}^{n}$ is a vector where the only non-null element equals one and is at the $i$-th row. $\Omega_{n,m}^{i,j} \in \mathbb{R}^{n \times m}$ is a single-entry matrix, where the only non-null element is equal to 1 and it is located at the $(i,j)$ entry. For $\mathbf{A} = [a_{ij}] \in \mathbb{R}^{n \times m}$ and $\mathbf{B} \in \mathbb{R}^{r \times l}$, $\mathbf{A} \otimes \mathbf{B} = [a_{ij} \mathbf{B}] \in \mathbb{R}^{r n \times m l}$ is the Kronecker product of $\mathbf{A}$ and $\mathbf{B}$. $\operatorname{det}(\mathbf{A})$ is the determinant of a square matrix $\mathbf{A}$. $|\cdot|$ is the absolute value. $\lambda_{i}(\mathbf{A})$ is the $i$-th eigenvalue of $\mathbf{A}$. $\mathcal{N} = \{1, 2, \dots, N \}$ denotes the set of followers in the platoon. $\underline{\lambda} = \underset{i \in \mathcal{N}}{\operatorname{min}} \, \lambda_{i}(\mathcal{\mathbf{A}}) $ and $ \overline{\lambda} = \underset{i \in \mathcal{N}}{\operatorname{max}} \, \lambda_{i}(\mathcal{\mathbf{A}}) $ denote the minimum and maximum eigenvalues of $\mathbf{A}$.



\subsection{Longitudinal Vehicle Dynamics}

	The vehicle dynamics considered here are based on \cite{Zheng_2016_Article,Zheng_2018_Article,Souza_2020_Article}. It assumes rigid and symmetrical vehicle bodies and negligible sliding of tires, pitch, and yaw moments. However, it considers the effects of wind speed, parametric uncertainty, and road inclination. Given this, the $i$-th vehicle nonlinear longitudinal dynamics is described by:
	\begin{align}
		      \dot{p}_{i}(t) \! &= \! v_{i}(t), \label{eq:LongDynamEq1} \\
		m a_{i}(t) \! &= \! {\frac{\eta }{r}} T_{i}(t) \! - \! \frac{1}{2} \rho c_{d} \bar{v}_{i}(t) | \bar{v}_{i}(t) | - m  g \sin{\theta(t)} \label{eq:LongDynamEq2} \\
		& - m g \mu \cos{\theta(t)} \bar{v}_{i}(t) / |\bar{v}_{i}(t)|, \nonumber
	\end{align}
	\noindent where $p_{i}(t)$, $v_{i}(t)$, and $a_{i}(t)$ are the position, velocity, and acceleration of the $i$-th vehicle, $m$ is the mass, $\eta$ is the drive-line mechanical efficiency, $r$ is the tire radius, $\rho$ is the air density, $c_{d}$ is the drag coefficient, $v_{w,i}(t)$ is the wind speed, $\bar{v}_{i}(t) \! = \! v_{i}(t) \! + \! v_{w,i}(t)$, $g$ is the acceleration due to gravity, $\mu$ is the rolling resistance coefficient, and $\theta(t)$ is the road inclination.
	
	The vehicle acceleration response, which relates the desired driving or braking torque $\widetilde{T}_{i}(t)$ with the actual torque $T_{i}(t)$, is modeled by the following first-order dynamics:
	\begin{equation}\label{eq:PowerTrainDynamics1stOrder}
		\varsigma \dot{T}_{i}(t) + T_{i}(t) = \widetilde{T}_{i}(t),
	\end{equation}
    \noindent where $\varsigma$ is the power-train time constant.
    
    A third-order linear model for the $i$-th vehicle is obtained by the feedback linearization technique, considering the following expression for the desired torque:
	\begin{align}
			\widetilde{T}_{i}(t) \! &= \! \frac{r}{\eta} \Bigl[ \frac{\rho c_{d}}{2} \Bigl( \!
		\bar{v}_{i}(t) | \bar{v}_{i}(t) | + \varsigma \dot{\bar{v}}_{i}(t) \bigl[ | \bar{v}_{i}(t) | + 2\bar{v}_{i}(t) / |\bar{v}_{i}(t)| \bigr] \! \Bigr) \nonumber \\ 
					   &  \qquad -\smash{m g \mu (\sin\theta(t)\dot{\theta}(t)\varsigma - \cos\theta(t) ) \bar{v}_{i}(t) / |\bar{v}_{i}(t)|} \nonumber \\ 
				       &  \qquad + m g (\cos\theta(t)\dot{\theta}(t)\varsigma + \sin\theta(t) ) + m u_{i}(t) \bigr], \label{eq:FeedbackLinearization}
	\end{align}
	\noindent where $u_{i}(t)$ is the desired input acceleration, and $\dot{\bar{v}}_{i}(t) \! = \! \dot{v}_{i}(t) \! + \! \dot{v}_{w,i}(t) $. The expression in \eqref{eq:FeedbackLinearization} is obtained by taking the derivative of \eqref{eq:LongDynamEq2} and substituting $\dot{T}_{i}(t)$, from \eqref{eq:PowerTrainDynamics1stOrder}, and $T_{i}(t)$ from \eqref{eq:LongDynamEq2}. Which simplifies to
	\begin{equation}\label{eq:AccLinearDynamics}
		\varsigma \dot{a}_{i}(t) = u_{i}(t) - a_{i}(t). \\
	\end{equation}
	Leading to the following state-space model:
    \begin{align}
        \dot{\mathbf{x}}_{i}(t) &= \mathbf{A} \mathbf{x}_{i}(t) + \mathbf{B} \left( u_{i}(t) + \psi_{i}(t) \right), \label{eq:platoonDynamicsTransition} 
    \end{align}
	\noindent where $\psi_{i}(t)$ represents an external disturbance, and $\mathbf{x}_{i}(t)$, $\mathbf{A}$, and $\mathbf{B}$ are given by:
    \begin{equation*}\label{eq:ModelMatrices}
        \arraycolsep=1.5pt\def\arraystretch{1}
        \mathbf{x}_{i}(t) \! = \!\! \left[\begin{array}{c}
            p_{i}(t) \\
            v_{i}(t) \\
            a_{i}(t)
            \end{array}\right]\!\!\!, \,\, \mathbf{A} \! = \!\! \left[\begin{array}{cccc}
            0 & 1 & 0 \\
            0 & 0 & 1 \\
            0 & 0 & \frac{-1}{\varsigma}
        \end{array}\right]\!\!\!, \,\, \mathbf{B} \! = \!\! \left[\begin{array}{c} 
            0 \\
            0 \\
            \frac{1}{\varsigma}
        \end{array}\right]. 
    \end{equation*}

\subsection{Network Topology}

    The platoon connection network is modeled as a directed graph $\mathcal{G}_{N+1} = \{ \mathcal{V}, \mathcal{E}, \mathcal{A} \}$, composed of a set of $N+1$ nodes $\mathcal{V} = \{\nu_{0}, \nu_{2}, \dots, \nu_{N}\}$, representing the vehicles, a set of edges $\mathcal{E} = \mathcal{V} \times \mathcal{V}$, representing a connection link. The adjacency matrix $\mathcal{A} = [a_{ij}] \in \mathbb{R}^{N \times N}$, indicates the connection link between the followers and is defined as
    \begin{equation}
    	a_{ij} = \begin{cases} 
    		1, & i \leftarrow j, \\ 
    		0, & i \nleftarrow j.
    	\end{cases} \quad
    	\arraycolsep=2pt\def\arraystretch{1}
    	\mathcal{A} = \! \begin{bmatrix}
    		0 & a_{1 2} & \cdots & a_{1 N} \\[0cm]
    		a_{2 1} &       0 & \cdots & a_{2 N} \\[-0.12cm]
    		\vdots &  \vdots & \ddots & \vdots  \\[-0.12cm]
    		a_{N 1} & a_{N 2} & \cdots & 0         
    	\end{bmatrix}, \label{eq:AdjacencyMatrix}
    \end{equation}
	\noindent where $i \leftarrow j$ means that node $i$ receives information from $j$. $a_{ij} = 0$ when that is not the case, and when $i=j$.
	
    Other characteristics of the graph are described by the Pinning, Degree, and Laplacian matrices. The Pinning matrix $\mathcal{P} = \operatorname{diag}\{p_{1}, p_{2}, \dots, p_{N}\} \in \mathbb{R}^{N \times N}$ indicates the connection link from the leader to the followers, defined as
        \begin{equation*}
        	p_{i} = \begin{cases} 
        		1, & i \leftarrow 0, \\ 
        		0, & i \nleftarrow 0.
        	\end{cases}
    \end{equation*}

    The Degree matrix $\mathcal{D} = \operatorname{diag}\{d_{1}, d_{2}, \dots, d_{N}\} \in \mathbb{R}^{N \times N}$, where $d_{i} = \sum_{k = 1}^{N} a_{ik}$, is defined as the in-degree of each node. The Laplacian matrix $\mathcal{L} = [l_{ij}] \in \mathbb{R}^{N \times N}$ is defined as $\mathcal{L} = \mathcal{D} - \mathcal{A}$, or equivalently as
    \begin{equation}
    	l_{ij} \!=\! \begin{cases} 
    		\sum\limits_{k = 1}^{N} a_{ik}, & \!\! i = j, \\
    		-a_{ij}, & \!\! i \neq j,
    	\end{cases} \quad
    	\arraycolsep=0.5pt\def\arraystretch{1}
            %
            %
    	\mathcal{L} \! = \! \begin{bmatrix}
    		   d_{1} & -a_{1 2} & \cdots & -a_{1 N} \\[0cm]
    		-a_{2 1} &    d_{2} & \cdots & -a_{2 N} \\[-0.12cm]
    		  \vdots &   \vdots & \ddots & \vdots   \\[-0.12cm]
    		-a_{N 1} & -a_{N 2} & \cdots & d_{N}         
    	\end{bmatrix}.
     \label{eq:LaplacianMatrix}
    \end{equation}

    The set of neighboring node vehicles connected to vehicle $i$ is denoted by $\mathbb{N}_{i} = \{ j \in \mathcal{N}| a_{ij} = 1 \}$. If the $i$-th vehicle is pinned to the leader, i.e., if it receives information from the leader, then $p_{i} = 1$, and $\mathbb{P}_{i} = \{ j = 0 | p_{i} = 1 \}$. Thus, $\mathbb{I}_{i} = \mathbb{N}_{i} \cup \mathbb{P}_{i}$ denotes the set of all neighboring vehicles connected to $i$.
    
    Some commonly used network topologies are shown in Figs.~\ref{fig:Topologies1} and~\ref{fig:Topologies2}. \ace{PF} and \ace{BD} topologies were often employed in the initial phases of platooning due to their dependency on radar-based communication \cite{Zheng_2016_Article}. In these configurations, a vehicle could only receive state information from its adjacent neighbors. 
    
    Advancements in wireless communication technology, particularly Dedicated Short Range Communication (DSRC), have facilitated the emergence of Vehicle-to-Vehicle (V2V) communication \cite{Gao_2018_Article}. This evolution has given rise to generalized topologies such as $r$PF(L) and $r$BD(L), where vehicles can actively exchange data with other vehicles in the platoon.
    
	\begin{figure}[htb]
		\centering
		\includegraphics[width=\linewidth]{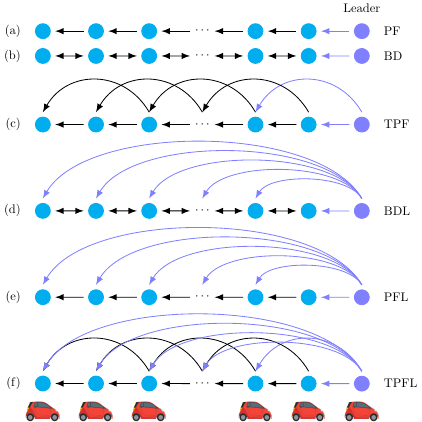}
		\caption{Network Topology: (a) \ace{PF}, (b) \ace{BD}, (c) Two PF (TPF), (d) BD $+$ Leader (BDL), (e) PF $+$ Leader (PFL), and (f) TPF $+$ Leader (TPFL).}
		\label{fig:Topologies1}
	\end{figure}

	\begin{figure}[htb]
		\centering
		\includegraphics[width=\linewidth]{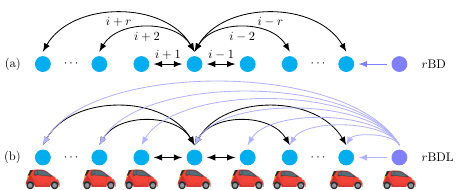}
		\caption{Generalized Undirected Network Topology: (a) $r$BD, and (b) $r$BD + Leader ($r$BDL).}
		\label{fig:Topologies2}
	\end{figure}
	
\subsection{Linear Distributed Control Law}
	
	The linear distributed control law in \eqref{eq:DistribControlLaw} was used to stabilize the platoon. Considering the relative state information of all the followers connected to the $i$-th  vehicle, and a \ace{CD} policy.
	
    \begin{equation}\label{eq:DistribControlLaw}
        \begin{aligned}
            u_i(t) &= -\sum_{j \in \mathbb{I}_i} \Biggl[ \kappa_{s} \int_{0}^{t} \left( p_{i}(t) - p_{j}(t) + d_{i j} \right) dt \Biggr. \\
            & \qquad\quad + \kappa_{p} \left( p_{i}(t) - p_{j}(t) + d_{i j} \right) \\
            & \qquad\quad + \smash{ \kappa_{v} \left( v_{i}(t) - v_{j}(t) \right) + \Biggl. \kappa_{a} \left( a_{i}(t) - a_{j}(t) \right) \biggr], }
        \end{aligned}
    \end{equation}
    
    \noindent where $\kappa_{\#}$ for $\# \in \{s, p, v, a\}$ are the controller gains and $d_{i j}$ is the front bumper to-rear-bumper constant desired distance.
    
    The control law in \eqref{eq:DistribControlLaw} is an improved version of the commonly used control protocols in the literature, with the addition of an integral term for the spacing error with gain $\kappa_{s}$. This simple modification leads the platoon to reach null steady-state spacing error by rejecting bounded constant disturbance inputs acting on each vehicle, such as road slope, wind speed, as well as parametric uncertainties in the longitudinal vehicle model.

\subsection{Closed-Loop Formation Error Dynamics}

	To analyze the internal stability and design the control gains, we extend the state vector $\mathbf{x}_{i}(t)$ to include a new state, $s_{i}(t) = \int_{0}^{t} p_{i}(t) dt $, into the vehicle model \eqref{eq:platoonDynamicsTransition}. Yielding the extended state-space model in \eqref{eq:ExtendedplatoonDynamicsTransition} for the $i$-th vehicle.
	\begin{align}
		\dot{\mathbf{x}}_{i}(t) &= \mathbf{A} \mathbf{x}_{i}(t) + \mathbf{B} u_i(t), \label{eq:ExtendedplatoonDynamicsTransition} 
	\end{align}
    \begin{equation*}\label{eq:ABxMatrices}
        \arraycolsep=0.5pt\def\arraystretch{1}
        \mathbf{x}_{i}(t) \! = \!\! \left[\begin{array}{c}
            s_{i}(t) \\
            p_{i}(t) \\
            v_{i}(t) \\
            a_{i}(t)
            \end{array}\right] \!\!\!, \,\, \arraycolsep=2pt\def\arraystretch{1} \mathbf{A} \! = \!\! \left[\begin{array}{cccc}
            0 & 1 & 0 & 0 \\
            0 & 0 & 1 & 0 \\
            0 & 0 & 0 & 1 \\
            0 & 0 & 0 & \frac{-1}{\varsigma}
	        \end{array}\right] \!\!\!, \,\, \arraycolsep=1pt\def\arraystretch{1} \mathbf{B} \! = \!\!\left[\begin{array}{c}
	        0 \\
	        0 \\
	        0 \\
	        \frac{1}{\varsigma} \end{array}\right] \!\!\!.
    \end{equation*}
	
	And the vector form of the control law \eqref{eq:DistribControlLaw} becomes
	\begin{equation}
		u_i(t) = -\mathbf{k}^{\top} \sum\nolimits_{j \in \mathbb{I}_i} \mathbf{x}_{i}(t) - \mathbf{x}_{j}(t) + \boldsymbol{\delta}_{ij}(t) = -\mathbf{k}^{\top} \boldsymbol{\epsilon}_{i}(t), \label{eq:DistribControlLawErrorVector}
	\end{equation}
	\noindent where $\mathbf{k} = \arraycolsep=2pt\def\arraystretch{1}\begin{bmatrix} \kappa_{s} & \kappa_{p} & \kappa_{v} & \kappa_{a} \end{bmatrix}^{\top}$ is the control gain vector, $\boldsymbol{\delta}_{ij}(t) = \arraycolsep=2pt\def\arraystretch{1}\begin{bmatrix} d_{i j} t & d_{i j} & 0 & 0 \end{bmatrix}^{\top}$ is the desired distance vector, and $\boldsymbol{\epsilon}_{i}(t)$ is the tracking error vector.
	
	Then, the augmented longitudinal dynamics of the platoon is formed by stacking each vehicle's state in augmented vectors, yielding the compact form in \eqref{eq:AugmentedplatoonDynamicsSS}.
    \begin{equation}\label{eq:AugmentedplatoonDynamicsSS}
        \dot{\mathbf{X}}(t) = \widehat{\mathbf{A}} \mathbf{X}(t) + \widehat{\mathbf{B}} \mathbf{U}(t), 
    \end{equation}
    \noindent where $\widehat{\mathbf{A}} = \mathbf{I}_{N} \otimes \mathbf{A}$ and $\widehat{\mathbf{B}} = \mathbf{I}_{N} \otimes \mathbf{B}$ are block diagonal matrices, and $\mathbf{X}(t) \!\! = \!\! \begin{bmatrix} \mathbf{x}_{1}(t) \!\! & \!\! \mathbf{x}_{2}(t) \!\! & \!\! \cdots \!\! & \!\! \mathbf{x}_{N}(t) \end{bmatrix}^{\top}$ and $\mathbf{U}(t) \!\! = \!\! \begin{bmatrix} u_{1}(t) \!\! & \!\! u_{2}(t) \!\! & \!\! \cdots \!\! & \!\! u_{N}(t) \end{bmatrix}^{\top}$ are the augmented state and input vectors.
	
    To compute the formation error dynamics, we first define a distance vector 
    $ \boldsymbol{\delta}_{0}^{\top} = 
    	\arraycolsep=2pt\def\arraystretch{1}
    	\left[\begin{array}{ccccc}
    		d_{0,1} & \cdots & d_{0,i} & \cdots & d_{0,N}
    	\end{array}\right] \in \mathbb{R}^{N \times 1} $, where $d_{0,i} = \sum_{j = 0}^{i-1} \left( d_{j,j+1} + l_{j} \right) - l_{0}$ is the front to rear bumper distance between the leader and the $i$-th follower, and $l_{j}$ is the length of the $j$-th vehicle.
    
    Following this, a target distance vector $\boldsymbol{\Delta}_{0}(t)$, representing the desired formation distance of each vehicle with respect to the leader, is computed as:
    \begin{equation*}
        \boldsymbol{\Delta}_{0}(t) = \boldsymbol{\delta}_{0} \otimes \left( t \mathbf{1}_{4}^{1} + \mathbf{1}_{4}^{2} \right) = \begin{bmatrix}
        						  	\boldsymbol{\delta}_{0,1} \!\! & \!\! \boldsymbol{\delta}_{0,2} \!\! & \!\! \cdots \!\! & \!\! \boldsymbol{\delta}_{0,N}
        						  \end{bmatrix}^{\top} \!\! \in \mathbb{R}^{4 N \times 1}, 
    \end{equation*}
	\noindent where $\arraycolsep=4pt\def\arraystretch{1}\boldsymbol{\delta}_{0,i} = \begin{bmatrix} d_{0,i} t & d_{0,i} & 0 & 0 \end{bmatrix} \in \mathbb{R}^{4 \times 1}$ .
	
	And by making $N$ copies of the leader's state vector $\mathbf{x}_{0}(t)$
    \begin{equation*}
    	\mathbf{X}_{0}(t) = \mathbf{1}_{N} \otimes \mathbf{x}_{0}(t) = \begin{bmatrix}
    		\mathbf{x}_{0}(t) \!\! & \!\! \mathbf{x}_{0}(t) \!\! & \!\! \cdots \!\! & \!\! \mathbf{x}_{0}(t)
    	\end{bmatrix}^{\top} \!\! \in \mathbb{R}^{4 N \times 1}, 
    \end{equation*}
    \noindent where $\arraycolsep=2.5pt\def\arraystretch{1}
    	\mathbf{x}_0(t) = \begin{bmatrix}
    		s_{0}(t) & p_{0}(t) & v_{0}(t) & a_{0}(t) \\
    	\end{bmatrix}^{\top} \in \mathbb{R}^{4 \times 1} $.
    
    Now, the formation error of the platoon can be defined as
    \begin{equation*}
    	\widetilde{\mathbf{X}}(t) \! = \! \mathbf{X}(t) - \mathbf{X}^{*}(t) \! = \! \begin{bmatrix}
    		\widetilde{\mathbf{x}}_{1}(t) \!\! & \!\! \widetilde{\mathbf{x}}_{2}(t) \!\! & \!\! \cdots \!\! & \!\! \widetilde{\mathbf{x}}_{N}(t)
    	\end{bmatrix}^{\top} \!\! \in \mathbb{R}^{4 N \times 1}, 
    \end{equation*}
    \noindent where $\mathbf{X}^{*}(t) = \mathbf{X}_{0}(t) - \boldsymbol{\Delta}_{0}(t)$ represents the desired reference trajectory, formed by the replicated leader's state $\mathbf{X}_{0}(t)$ and its relative distance from each follower, $\boldsymbol{\Delta}_{0}(t)$.
    
    Thus, for a virtual leader, the formation error dynamics is
    \begin{align}
    	\dot{\widetilde{\mathbf{X}}}(t) &= \mathbf{\dot{X}}(t) - \mathbf{\dot{X}}^{*}(t), \nonumber \\
    									&= \dot{\mathbf{X}}(t) - \dot{\mathbf{X}}_{0}(t) + \dot{\boldsymbol{\Delta}}_{0}(t), \nonumber \\
										&= \dot{\mathbf{X}}(t) - \widehat{\mathbf{A}}\mathbf{X}_{0}(t) + \dot{\boldsymbol{\Delta}}_{0}(t), \label{eq:FormationErrorDynPartI}
    \end{align}
    \noindent where $\dot{\mathbf{X}}_{0}(t) = \widehat{\mathbf{A}}\mathbf{X}_{0}(t) + \widehat{\mathbf{B}}\mathbf{U}_{0}(t) = \widehat{\mathbf{A}}\mathbf{X}_{0}(t)$, given the fact that $\epsilon_{0}(t) = 0$, the tracking error of the leader with respect to itself is null, then $\mathbf{U}_{0}(t) = 0$.
        
    By replacing $\mathbf{X}_{0}(t) = \mathbf{X}^{*}(t) + \boldsymbol{\Delta}_{0}(t)$ in \eqref{eq:FormationErrorDynPartI}, and noting that in \eqref{eq:DeltazeroDot}, $\mathbf{A}\boldsymbol{\delta}_{0,i}(t) = \dot{\boldsymbol{\delta}}_{0,i} \implies \widehat{\mathbf{A}} \boldsymbol{\Delta}_{0}(t) = \dot{\boldsymbol{\Delta}}_{0}(t)$, then
    \begin{align}
    	\dot{\widetilde{\mathbf{X}}}(t) &= \dot{\mathbf{X}}(t) - \widehat{\mathbf{A}}\mathbf{X}^{*}(t) - \widehat{\mathbf{A}} \boldsymbol{\Delta}_{0}(t) + \dot{\boldsymbol{\Delta}}_{0}(t), \label{eq:DeltazeroDot} \\
									    &= \dot{\mathbf{X}}(t) - \widehat{\mathbf{A}}\mathbf{X}^{*}(t). \label{eq:FormationErrorDynPartII}
	\end{align}
	
	Then, substituting \eqref{eq:AugmentedplatoonDynamicsSS} in \eqref{eq:FormationErrorDynPartII} and $ \mathbf{X}(t) \! - \! \mathbf{X}^{*}(t)$ with $\widetilde{\mathbf{X}}(t)$, we obtain an expression for the formation error dynamics of the entire platoon:
	\begin{align}
		\dot{\widetilde{\mathbf{X}}}(t) &= \widehat{\mathbf{A}}\mathbf{X}(t) + \widehat{\mathbf{B}}\mathbf{U}(t)  - \widehat{\mathbf{A}}\mathbf{X}^{*}(t), \label{eq:FormationErrorDynPartIII} \\ 
     &= \widehat{\mathbf{A}} \bigl(  \mathbf{X}(t) - \mathbf{X}^{*}(t) \bigr) + \widehat{\mathbf{B}}\mathbf{U}(t), \nonumber \\
  &= \widehat{\mathbf{A}}\widetilde{\mathbf{X}}(t)+\widehat{\mathbf{B}}\mathbf{U}(t).\label{eq:FormationErrorDynPartFinal}
    \end{align}
	
	To express the augmented input vector $\mathbf{U}(t)$ in terms of the formation error $\mathbf{\widetilde{X}}(t)$ in a compact form, we expand $\boldsymbol{\epsilon}_{i}(t)$ in \eqref{eq:DistribControlLawErrorVector} for each node vehicle $j$ connected to $i$. By using the network topology information given by $\mathbb{I}_{i} = \mathbb{N}_{i} \cup \mathbb{P}_{i} = \{ j \in \mathcal{N}| a_{ij} = 1 \} \cup \{ j = 0 | p_{i} = 1 \}$.
	\begin{align}
		\boldsymbol{\epsilon}_{i}(t) \! &= \! \sum\nolimits_{j \in \mathbb{N}_{i} \cup \mathbb{P}_{i}} \mathbf{x}_{i}(t) - \mathbf{x}_{j}(t) + \boldsymbol{\delta}_{i j}, \\
						  &= \! \sum\nolimits_{j=1}^{N} \!\!\! a_{ij} \! \left( \mathbf{x}_{i}(t) \! - \! \mathbf{x}_j(t) \! + \! \boldsymbol{\delta}_{i j} \! \right) \! + \! p_{i} \! \left( \mathbf{x}_{i}(t) \! - \! \mathbf{x}_{0}(t) \! + \! \boldsymbol{\delta}_{i 0} \right), \nonumber \\
						  &= \! \sum\nolimits_{j=1}^{N} a_{ij} \left( \mathbf{\widetilde{x}}_{i}(t) - \mathbf{\widetilde{x}}_j(t) \right) + p_{i} \mathbf{\widetilde{x}}_{i}(t), \label{eq:CompactUExpansionPartI} \\
						  &= \! \Bigl( p_{i} + \sum\nolimits_{j=1}^{N} a_{i j} \Bigr) \mathbf{\tilde{x}}_{i}(t) - \sum\nolimits_{j = 1, j \neq i}^{N} a_{i j} \mathbf{\tilde{x}}_{j}(t). \label{eq:CompactUExpansionPartII}
	\end{align}

	To obtain \eqref{eq:CompactUExpansionPartI} we used the fact that $\boldsymbol{\delta}_{i j}(t) = \boldsymbol{\delta}_{i 0}(t) - \boldsymbol{\delta}_{0 j}(t)$, then by adding and subtracting $\mathbf{x}_{0}(t)$ in the first term, and using $\mathbf{x}^{*}_{i}(t) \! = \! \mathbf{x}_{0}(t) \! - \! \boldsymbol{\delta}_{i 0}(t)$ to substitute $\mathbf{x}_{i}(t) \! - \! \mathbf{x}^{*}_{i}(t) $ by $ \mathbf{\widetilde{x}}_{i}(t)$. 
	
	Stacking \eqref{eq:CompactUExpansionPartII} $\forall i \in \mathcal{N}$ and isolating $\mathbf{\widetilde{X}}(t)$, we obtain the compact form of the control law 
 
	\begin{align}
		\mathbf{U}(t) \! &= \! -\mathbf{k}^{\top} \Bigl[  \sum_{i = 1}^{N} \Omega_{N}^{i,i} \left( p_{i} + d_{i} \right) \mathbf{\tilde{x}}_{i}(t) \! - \! \sum_{i = 1}^{N} \sum_{j=1}^{N} \Omega_{N}^{i,j} a_{i j} \mathbf{\tilde{x}}_{j}(t) \Bigr], \nonumber \\ 
					  &= - \left\{ \!\! \left[ \sum_{i = 1}^{N} \Omega_{N}^{i,i} \! \left( \! \left(p_{i} + d_{i}\right) \! - \! \sum_{j=1}^{N} \Omega_{N}^{i,j} a_{i j} \right) \! \right] \!\! \otimes \mathbf{k}^{\top} \! \right\} \! \mathbf{\widetilde{X}}(t), \nonumber \\
					  &= - \left[ \left( \mathcal{L} + \mathcal{P} \right) \otimes \mathbf{k}^{\top} \right] \mathbf{\widetilde{X}}(t). \label{eq:CompactFormationErrorFinal}
	\end{align}

	Finally, the closed-loop form of the platoon formation error in \eqref{eq:ClosedLoopFormationErrorDynamics} is obtained by substituting \eqref{eq:CompactFormationErrorFinal} into \eqref{eq:FormationErrorDynPartFinal}:
	\begin{equation}
		\dot{\widetilde{\mathbf{X}}}(t) = \bigl\{ \widehat{\mathbf{A}} - \left( \mathcal{L} + \mathcal{P} \right) \otimes \mathbf{B} \mathbf{k}^{\top} \bigr\} \widetilde{\mathbf{X}}(t) = \widetilde{\mathbf{A}}_{c} \widetilde{\mathbf{X}}(t), \label{eq:ClosedLoopFormationErrorDynamics}
	\end{equation}
	\noindent where $\widetilde{\mathbf{A}}_{c} = \widehat{\mathbf{A}} - \left( \mathcal{L} + \mathcal{P} \right) \otimes \mathbf{B} \mathbf{k}^{\top}$.

\section{Platoon Internal Stability Analysis}\label{sec:StabilityAnalysis}
	
This section provides the stability analysis of the formation error dynamics of the platoon. First, stability conditions are derived for topologies of look-ahead type, also referenced as $r$PFL. Then, the analysis is extended to general undirected topologies, generically referenced as $r$BDL. Before beginning, we state the following assumptions:
\begin{assumption}[Assumption 3.1 \cite{Lewis_2013_Book}]\label{assump:SpanningTree}
	The directed or undirected augmented graph $\mathcal{G}_{N+1}$ contains a spanning tree rooted at the leader node.
\end{assumption}

\begin{assumption}\label{assump:Controlability}
	The pair $(\mathbf{A},\mathbf{B})$ is controllable.
\end{assumption}


%


\subsection{\textbf{Stability Analysis for Look-Ahead Topologies}}

\begin{theorem}[Asymptotic Stability for $r$PFL Topologies]\label{thm:LookAheadStability}
	 Consider a homogeneous platoon with dynamics \eqref{eq:platoonDynamicsTransition}, regulated by the control law \eqref{eq:DistribControlLaw}, and connected through a known and fixed $r$PFL topology. The following statements hold:
	 \begin{enumerate}
	 	\item Every follower is connected to one or more node vehicles: $n_{i} = p_{i} + d_{i} > 0, \quad n_{i} \in \mathbb{Z}^{+}, \quad \forall i \in \mathcal{N}$
	 	\item The formation error dynamics in \eqref{eq:ClosedLoopFormationErrorDynamics} is asymptotically stable if and only if
	 		\begin{align}
	 			& \kappa_{s} > 0, \quad 0 < \kappa_{p} < \kappa_{v} \bigl( 1 + \overline{n} \kappa_{a} \bigr)/\varsigma, \quad \kappa_{a} > -1/\overline{n}, \nonumber \\
	 			& \kappa_{v} > \bigl( \kappa_{s} \bigl( 1 + \overline{n} \kappa_{a} \bigr)^2 + \varsigma \overline{n} (\kappa_{p})^2 \bigr) / \underline{n} \bigl( 1 + \underline{n} \kappa_{a} \bigr) \kappa_{p} \label{eq:ControlGainsLookAhead}
	 		\end{align} 		
	 \end{enumerate}
	\noindent where $\underline{n} = \underset{i \in \mathcal{N}}{\operatorname{min}} \, n_{i}(\mathcal{P}+\mathcal{D}) $ and $ \overline{n} = \underset{i \in \mathcal{N}}{\operatorname{max}} \, n_{i}(\mathcal{P}+\mathcal{D}) $.
 
\end{theorem}

\begin{proof}
	
	The formation error dynamics is asymptotically stable if and only if \eqref{eq:ClosedLoopTransitionMatrix} is Hurwitz.
	\begin{equation}
		\widetilde{\mathbf{A}}_{c} = \mathbf{I}_{N} \otimes \mathbf{A} - \left( \mathcal{L} + \mathcal{P} \right) \otimes \mathbf{B} \mathbf{k}^{\top} \label{eq:ClosedLoopTransitionMatrix}
	\end{equation}

	Requiring that all the roots of the characteristic equation in \eqref{eq:CharEqLookAhead} have negative real parts, which implies that
	\begin{equation}
		\Phi(s) = \operatorname{det} \! \left( s \mathbf{I}_{4N} - \widetilde{\mathbf{A}}_{c} \right) \label{eq:CharEqLookAhead}
	\end{equation}
	
	$\widehat{\mathbf{A}}$ is a block diagonal matrix by definition. For $r$PFL topologies, $\mathcal{L} + \mathcal{P}$ is strictly lower triangular. In this case, its eigenvalues are $n_{i} = p_{i} + d_{i} \in \mathbb{Z}^{+} $.
	
	Assumption \ref{assump:SpanningTree} guarantees that $n_{i} > 0, \forall i \in \mathcal{N}$. Making the leader's state information available directly or indirectly to all the followers, proving statement 1.

	To prove statement 2, we noticed that according to statement 1, \eqref{eq:CharEqLookAhead} is equivalent to \eqref{eq:EquivalentCharEqLookAheadPartI} \cite{Godinho_2022_Article}, with the $i$-th characteristic equation in \eqref{eq:EquivalentCharEqLookAheadPartII}.
	\begin{equation}
		\Phi(s) = \prod_{i = 1}^{N} \operatorname{det} \bigl( s \mathbf{I}_{4} - \mathbf{A} +  n_{i} \mathbf{B} \mathbf{k}^{\top} \bigr) = \prod_{i = 1}^{N} \phi_{i}(s) \label{eq:EquivalentCharEqLookAheadPartI}
	\end{equation}
	\begin{equation}\label{eq:EquivalentCharEqLookAheadPartII}
		\phi_{i}(s) \! = \! s^{4} \! + \! \left( \! \frac{1 + n_{i} \kappa_{a}}{\varsigma} \! \right) \!\! s^{3} + \frac{n_{i} \kappa_{v}}{\varsigma} s^{2} + \frac{n_{i} \kappa_{p}}{\varsigma} s + \frac{n_{i} \kappa_{s}}{\varsigma}
	\end{equation}

	Then, the Routh-Hurwitz stability criteria yields
	\begin{equation*}\label{eq:RouthHurwitzCriteriaLookAhead}
		\begin{NiceArray}{l|ccc}
			s^{4} & 1                    & \frac{n_{i} \kappa_{v}}{\varsigma} & \frac{n_{i} \kappa_{s}}{\varsigma} \\
			s^{3} & \frac{1 + n_{i} \kappa_{a}}{\varsigma} & \frac{n_{i} \kappa_{p}}{\varsigma} \\
			s^{2} & \alpha_{i}           & \frac{n_{i} \kappa_{s}}{\varsigma} \\
			s^{1} & \beta_{i}            &                    \\
			s^{0} & \frac{n_{i} \kappa_{s}}{\varsigma}                & \\
			\end{NiceArray} \quad \begin{aligned}
			\alpha_{i} \! &= \! \frac{\bigl( 1 \! + \! n_{i} \kappa_{a} \bigr) \bigl( n_{i} \kappa_{v} \bigr) \!\! - \! n_{i} \kappa_{p}}{ \varsigma \bigl( 1 \! + \! n_{i} \kappa_{a} \bigr)} \\
			\beta_{i} \! &= \! \frac{n_{i} \kappa_{p}}{\varsigma} - \frac{n_{i} \kappa_{s} \bigl(1 \! + \! n_{i} \kappa_{a} \bigr) }{\alpha_{i} \varsigma^{2} }
		\end{aligned}
	\end{equation*}	
	\begin{align*}
		& \kappa_{s} > 0, \quad 0 < \kappa_{p} < \kappa_{v} \bigl( 1 + n_{i} \kappa_{a} \bigr) / \varsigma, \quad \kappa_{a} > -1/n_{i}, \nonumber \\
		& \kappa_{v} > \bigl( \kappa_{s} \bigl( 1 + n_{i} \kappa_{a} \bigr)^2 + \varsigma n_{i} (\kappa_{p})^2 \bigr) / n_{i} \bigl( 1 + n_{i} \kappa_{a} \bigr) \kappa_{p}
	\end{align*}
	We know that $\varsigma > 0$ and $n_{i}> 0$, $\forall i \in \mathcal{N}$. The gains $\kappa_{s}$ and $\kappa_{a}$ are obtained directly from the table by inspection. Whereas $\kappa_{v}$ and $\kappa_{p}$ are computed by imposing $\alpha_{i}, \beta_{i} > 0$, and solving for the respective gains. Therefore, $\mathbf{A} - n_{i} \mathbf{B} \mathbf{k}^{\top} $ is asymptotically stable if and only if \eqref{eq:ControlGainsLookAhead} is satisfied.
\end{proof}

\begin{corollary}[Asymptotic Stability for $\kappa_{s} = 0$ and $r$PFL Topologies]\label{cor:LookAheadKappaSZero}

	It follows from Theorem \ref{thm:LookAheadStability} that if $\kappa_{s} = 0$, then the original system in \eqref{eq:platoonDynamicsTransition}, without the integrator state, is asymptotically stable if and only if \eqref{eq:CorollarykKappaSZeroGains} is satisfied.
	\begin{equation}
		\kappa_{p} > 0, \quad \kappa_{v} > \varsigma \kappa_{p} / \bigl( 1 + \underline{n} \kappa_{a} \bigr), \quad \kappa_{a} > -1/\overline{n} \label{eq:CorollarykKappaSZeroGains}
	\end{equation}
\end{corollary}

\begin{proof}
	
	The proof is straightforward. Given the modified characteristic polynomial in \eqref{eq:EquivalentCharEqLookAheadKappaSZero} and the Routh-Hurwitz stability analysis in \eqref{eq:RouthHurwitzCriteriaLookAheadKappaSZero}.
	\begin{equation}\label{eq:EquivalentCharEqLookAheadKappaSZero}
		\bar{\phi}_{i}(s) = \frac{\phi_{i}}{s} \biggr\rvert_{\kappa_{s}=0} \!\!\! = \! s^{3} \! + \! \left( \! \frac{1 + n_{i} \kappa_{a}}{\varsigma} \! \right) \!\! s^{2} + \frac{n_{i} \kappa_{v}}{\varsigma} s + \frac{n_{i} \kappa_{p}}{\varsigma}
	\end{equation}
	\begin{equation}\label{eq:RouthHurwitzCriteriaLookAheadKappaSZero}
		\begin{NiceArray}{l|cc}
			s^{3} & 1 & \frac{n_{i} \kappa_{v}}{\varsigma}  \\
			s^{2} & \frac{1 + n_{i} \kappa_{a}}{\varsigma} & \frac{n_{i} \kappa_{p}}{\varsigma} \\
			s^{1} & \alpha_{i} \\
			s^{0} & \frac{n_{i} \kappa_{s}}{\varsigma} \\
		\end{NiceArray}
	\end{equation}
\end{proof}	
	
 Although stability is preserved, null steady-state spacing error is no longer guaranteed. This result is in line with the stabilizing thresholds found in \cite{Godinho_2022_Article, Neto_2019_Article, Godinho_2023_Article}. Also, Cor.~\ref{cor:LookAheadKappaSZero} reveals that \cite{Zheng_2021_Article} becomes a particular case of Thm.~\ref{thm:LookAheadStability} when $\kappa_{s} = 0$.

\subsection{\textbf{Stability Analysis for Undirected Topologies}}

The stability analysis for the undirected case requires the following Proposition and Lemmas.

\begin{lemma}[Theorem 6.1.1: Geršgorin Disk Criteria \cite{Horn_2012_Book}]\label{lem:GershogorinDiskCriteria}
	Let a square matrix $\mathcal{M} = [m_{ij}] \in \mathbb{R}^{n \times n}$, all its eigenvalues lie in the union of $n$ disks
	\begin{equation}
		S(\mathcal{M}) \subset \bigcup_{i=1}^{n} \left\{ \lambda \in \mathbb{C} \,\, \bigl. \Bigr| \,\, | \lambda - m_{ii} | \leqslant \!\!\!\! \sum_{j=1,j \neq i}^{n} \!\!\! \left| m_{ij} \right| \right\}. 
	\end{equation}
\end{lemma}

\begin{proposition}[Proposition 3.10 \cite{Mesbahi_2010_Book}]\label{propos:LPSpectrum}
	The spectrum of $\mathcal{M} = \mathcal{L} + \mathcal{P}$ lies in the region
	\begin{equation}
		S(\mathcal{M}) = \Bigl\{ \lambda \in \mathbb{C}  \,\, \bigl. \bigr| \,\, \bigl| \lambda - \underset{i \in \mathcal{N}}{\operatorname{max}} \, \lambda_{i}(\mathcal{\mathcal{M}}) \bigr| \leqslant \underset{i \in \mathcal{N}}{\operatorname{max}} \, \lambda_{i}(\mathcal{\mathcal{M}}) \Bigr\}
	\end{equation}
	Then $\operatorname{Re}\{\lambda_{i}(\mathcal{M})\} \geqslant 0$ $\forall i \in \mathcal{N}$, i.e., the eigenvalues of $\mathcal{M}$ have non-negative real parts.
\end{proposition}

\begin{lemma}[Theorem \cite{Shivakumar_1974_Article}]\label{lem:LPNonsingular}
	Let a matrix $\mathcal{M} = [m_{ij}] \in \mathbb{R}^{n \times n}$ and
	\begin{equation}
		\mathbb{J} = \left\{ i \in \mathcal{N} \,\, \bigl. \Bigr| \,\, |m_{ii} | > \!\!\!\! \sum_{j=1, j \neq i}^{n} \!\!\! | m_{ij} | \right\} \neq \emptyset .
	\end{equation}
	Then $\mathcal{M}$ is non-singular if $\forall i \notin \mathbb{J}$ $\exists$ $\left\{m_{i i_1}, m_{i_1 i_2}, \dots, m_{i_r j}\right\} \neq \emptyset $ with $j \in \mathbb{J}$.
\end{lemma}

\begin{lemma}[Lemma 3.4 \cite{Lewis_2013_Book}/ Theorem 3 \cite{OlfatiSaber_2004_Article}]\label{lem:JordanDecomposition}
	Let $\lambda_{i} = \lambda_{i}(\mathcal{\mathcal{L} + \mathcal{P}})$ $\forall i \in \mathcal{N}$, distinct or repeated. The platoon dynamics \eqref{eq:ClosedLoopFormationErrorDynamics} is asymptotically stable if and only if all the following matrices are Hurwitz
	\begin{equation}
		\bar{\mathbf{A}}_{c_{i}} = \mathbf{A} - \lambda_{i} \mathbf{B} \mathbf{k}^{\top}, \,\, \forall i \in \mathcal{N}
	\end{equation}
\end{lemma}

We are now ready to state the stability thresholds for the undirected topologies.

\begin{theorem}[Asymptotic Stability for $r$BDL Topologies]\label{thm:BidirectionalStability}
	
	Consider a homogeneous platoon with longitudinal dynamics \eqref{eq:platoonDynamicsTransition}, ruled by the control law \eqref{eq:DistribControlLaw}, and connected through a known and fixed undirected communication topology. The following statements hold:
	\begin{enumerate}
		\item $\mathcal{M} = \mathcal{L} + \mathcal{P}$ is non-singular and all its eigenvalues are real: $\lambda_{i} \in \mathbb{R}^{+}$, $\forall i \in \mathcal{N}$.
		\item The formation error dynamics in \eqref{eq:ClosedLoopFormationErrorDynamics} is asymptotically stable if and only if
		\begin{align}
			& \kappa_{s} > 0, \quad 0 < \kappa_{p} < \kappa_{v} \bigl( 1 + \overline{\lambda} \kappa_{a} \bigr) / \varsigma, \quad \kappa_{a} > -1/\overline{\lambda}, \nonumber \\
			& \kappa_{v} > \bigl( \kappa_{s} \bigl( 1 + \overline{\lambda} \kappa_{a} \bigr)^2 + \varsigma \overline{\lambda} (\kappa_{p})^2 \bigr) \underline{\lambda} \bigl( 1 + \underline{\lambda} \kappa_{a} \bigr) \kappa_{p} \label{eq:ControlGainsUndirected}
		\end{align}
	\end{enumerate}

\end{theorem}

\begin{proof}[Proof of Statement 1]
	
	From the Laplacian matrix definition in \eqref{eq:LaplacianMatrix}, $l_{ii} \geqslant 0$, and $l_{ij} \leqslant 0$. Also, since $ \sum_{j = 1}^{N} l_{ij} = 0 $, $\forall i \in \mathcal{N}$:
	\begin{equation}\label{eq:Theorem2ProofStatement1}
		| l_{ii} | = \!\!\!\! \sum_{j=1, j \neq i}^{N} \!\!\! | l_{ij} | \geqslant 0 \implies | l_{ii} + p_{i} | \geqslant \!\!\!\! \sum_{j=1, j \neq i}^{N} \!\!\! | l_{ij} |, \,\, i \in \mathcal{N}
	\end{equation}
	\noindent with $p_{i} > 0$ for at least one $i \in \mathcal{N}$, under Assumption \ref{assump:SpanningTree}.
	Using Lemma \ref{lem:GershogorinDiskCriteria}, all the eigenvalues of $\mathcal{M} = \mathcal{L} + \mathcal{P}$ are contained in the union of $N$ disks, centered at $m_{ii} = l_{ii} + p_{i}, $ with radius $r_{i} = \sum_{j=1, j \neq i}^{N} | l_{ij} |$:
	\begin{equation}\label{eq:GersgorinDiskUnion}
		S(\mathcal{M}) \subset \bigcup_{i=1}^{N} \left\{ \lambda \in \mathbb{C} \,\, \bigl. \Bigr| \,\, | \lambda - \bigl( l_{ii} + p_{i} \bigr) | \leqslant r_{i} \right\}.
	\end{equation}

	According to Proposition \ref{propos:LPSpectrum}, the spectrum of $\mathcal{M} = \mathcal{L} + \mathcal{P}$ lies in the range $S(\mathcal{M}) = \Bigl\{ \lambda \in \mathbb{C} \,\, \bigl. \bigr| \,\, \bigl| \lambda - \overline{\lambda} \bigr| \leqslant \overline{\lambda} \Bigr\}$. Thus, implying that $\operatorname{Re}\{\lambda_{i}(\mathcal{M})\} \geqslant 0$ $\forall i \in \mathcal{N}$.
	
	From Lemma \ref{lem:LPNonsingular}, the non-singularity of $\mathcal{M}$ also implies that all its eigenvalues are located at the open right-half plane. Assuming, without loss of generality, that the first $r$ vehicles are pinned to the leader, $p_{1}, p_{2}, \dots, p_{r} > 0$, implies that
	\begin{equation}\label{eq:Theorem2ProofStatement1Lemma2}
		| l_{ii} + p_{i} | > \!\!\!\! \sum_{j=1, j \neq i}^{N} \!\!\! | l_{ij} |, \qquad i \in \{1,2, \dots, r \}
	\end{equation}
	Then $\mathbb{J} \neq \emptyset$, and the strict inequality in \eqref{eq:Theorem2ProofStatement1Lemma2} holds.
	
	To conclude the proof of statement 1, note that for undirected topologies, the symmetry of $\mathcal{M}$ guarantees that its eigenvalues are strictly real-valued numbers. For directed topologies, $\mathcal{M}$ is a lower triangular matrix, and its eigenvalues are integer numbers, $n_{i} = \lambda_{i} = l_{ii} + p_{i}$. 
\end{proof}

\begin{proof}[Proof of Statement 2]
	
	To prove statement 2 we begin by proving Lemma \ref{lem:JordanDecomposition} by computing a similarity transformation of $\widetilde{\mathbf{A}}_{c}$ in \eqref{eq:ClosedLoopFormationErrorDynamics}, with $\mathbf{T} \in \mathbb{R}^{N \times N}$ a nonsingular matrix:
	\begin{align}
		\bar{\mathbf{A}}_{c} &= \bigl( \mathbf{T} \otimes \mathbf{I}_{4} \bigr)^{-1} \bigl( \mathbf{I}_{N} \otimes \mathbf{A} - \left( \mathcal{L} + \mathcal{P} \right) \otimes \mathbf{B} \mathbf{k}^{\top} \bigr) \bigl( \mathbf{T} \otimes \mathbf{I}_{4} \bigr) \nonumber \\
							 &= \mathbf{I}_{N} \otimes \mathbf{A} - \bigl[ \mathbf{T}^{-1} \bigl( \mathcal{L} + \mathcal{P} \bigr) \mathbf{T} \bigr] \otimes \mathbf{B} \mathbf{k}^{\top} \nonumber \\
							 &= \mathbf{I}_{N} \otimes \mathbf{A} - \mathbf{J} \otimes \mathbf{B} \mathbf{k}^{\top} \label{eq:DiagonalTransformation}
	\end{align}
	\noindent where $\mathbf{J} = \mathbf{T}^{-1} \bigl( \mathcal{L} + \mathcal{P} \bigr) \mathbf{T} = \operatorname{diag}\{ \mathbf{J}_{n_{1}}^{(\lambda_{1})}, \mathbf{J}_{n_{2}}^{(\lambda_{2})}, \dots, \mathbf{J}_{n_{k}}^{(\lambda_{k})} \}$, with $\mathbf{J}_{n_{i}}^{(\lambda_{i})} \in \mathbb{R}^{n_{i} \times n_{i}}$ a Jordan block with $\lambda_{i}$ on its main diagonal, and $n_{1} + n_{2} + \cdots + n_{k} = N$.
	
	Then $\bar{\mathbf{A}}_{c}$ is Hurwitz if and only if its block diagonal entries $\bar{\mathbf{A}}_{c_{i}} = \mathbf{A} - \lambda_{i} \mathbf{B} \mathbf{k}^{\top}$ are Hurwitz, $\forall i \in \mathcal{N} $. Given that the eigenvalues of a block diagonal matrix are the union of the eigenvalues of the diagonal blocks \cite{Lewis_2013_Book}. The characteristic equation of $\bar{\mathbf{A}}_{c_{i}}$ is given by \eqref{eq:EquivalentCharEqundirectedPartI} and \eqref{eq:EquivalentCharEqUndirectedPartII}.
	\begin{equation}
		\Phi(s) = \prod_{i = 1}^{N} \operatorname{det} \! \left( s \mathbf{I}_{4} - \mathbf{A} + \lambda_{i} \mathbf{B} \mathbf{k}^{\top} \right) = \prod_{i = 1}^{N} \phi_{i}(s) \label{eq:EquivalentCharEqundirectedPartI}
	\end{equation}
	\begin{equation}\label{eq:EquivalentCharEqUndirectedPartII}
		\phi_{i}(s) \! = \! s^{4} \! + \! \left( \! \frac{1 + \lambda_{i} \kappa_{a}}{\varsigma} \! \right) \!\! s^{3} + \frac{\lambda_{i} \kappa_{v}}{\varsigma} s^{2} + \frac{\lambda_{i} \kappa_{p}}{\varsigma} s + \frac{\lambda_{i} \kappa_{s}}{\varsigma}
	\end{equation}
		
	
		Similarly to Theorem \ref{thm:LookAheadStability}, given that $\varsigma > 0$ and $\lambda_{i} > 0$, $\forall i \in \mathcal{N}$, the Routh-Hurwitz criteria yields:
	\begin{align*}
		& \kappa_{s} > 0, \quad 0 < \kappa_{p} < \kappa_{v} \bigl( 1 + \lambda_{i} \kappa_{a} \bigr) / \varsigma, \quad \kappa_{a} > -1/\lambda_{i}, \nonumber \\
		& \kappa_{v} > \bigl( \kappa_{s} \bigl( 1 + \lambda_{i} \kappa_{a} \bigr)^2 + \varsigma \lambda_{i} (\kappa_{p})^2 \bigr) / \lambda_{i} \bigl( 1 + \lambda_{i} \kappa_{a} \bigr) \kappa_{p}
	\end{align*}
	Therefore, the platoon dynamics $\mathbf{A} - \lambda_{i} \mathbf{B} \mathbf{k}^{\top} $ is asymptotically stable if and only if \eqref{eq:ControlGainsUndirected} is fulfilled. 
\end{proof}

\begin{corollary}[Asymptotic Stability for $\kappa_{s} = 0$ and $r$BDL Topologies]\label{cor:BidirectionalKappaSZero}
	
	Similar to Corollary \ref{cor:LookAheadKappaSZero}, asymptotically stability is ensured if $\kappa_{s} = 0$ for undirected topologies. In this case, the system in \eqref{eq:platoonDynamicsTransition} is asymptotically stable if and only if \eqref{eq:CorollarykKappaSZeroGainsUndirected} is satisfied.
	\begin{equation}
		\kappa_{p} > 0, \quad \kappa_{v} > \varsigma \kappa_{p} / \bigl( 1 + \underline{\lambda} \kappa_{a} \bigr), \quad \kappa_{a} > -1/\overline{\lambda} \label{eq:CorollarykKappaSZeroGainsUndirected}
	\end{equation}
\end{corollary}

\begin{proof}
	The proof follows the same lines of Cor.~\ref{cor:LookAheadKappaSZero}, simply replacing $n_{i}$ with $\lambda_{i}$ in the characteristic polynomial in \eqref{eq:EquivalentCharEqLookAheadKappaSZero}-\eqref{eq:RouthHurwitzCriteriaLookAheadKappaSZero}.
\end{proof}

\begin{remark}[Null steady-state error for $\kappa_{s} \neq 0$]\label{rem:KappaSEffectZeroSSError}
	
	Consider the closed-loop system for the $i$-th vehicle in Fig.~\ref{fig:blockdiagram}.
	
	\begin{figure}[htb]
		\centering
		\includegraphics[width=\linewidth]{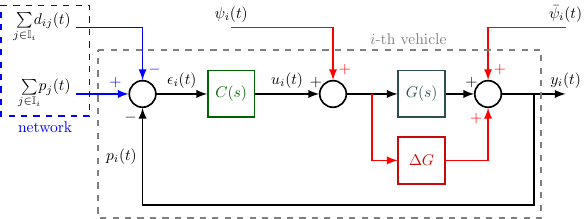}
		\caption{Closed-loop block diagram for the $i$-th vehicle.} 
	\label{fig:blockdiagram}
\end{figure}

Defining $\mathbf{C} = \operatorname{diag}\{c_{s}, c_{p}, c_{v}, c_{a}\}$ with $c_{\#} \in \{0,1\}$ for $\# \in \{p, v, a\}$, and denoting by $s$ the Laplace transform variable. We feed-back only the position, i.e., only $c_{p} = 1$, then $ y_{i}(t) = \mathbf{C} \mathbf{x}_{i}(t) = p_{i}(t) $, and the tracking spacing error becomes:
\begin{align}\label{eq:TrackingSpacingError}
	\epsilon_{i}(t) \! &= \! - \! \sum_{j \in \mathbb{I}_i} \! y_{i}(t) \! - \! y_{j}(t) \! + \! d_{ij}(t) \! = \! - p_{i}(t) \! + \!\! \sum_{j \in \mathbb{I}_i} p_{j}(t) \! - \! d_{ij}(t) \nonumber
\end{align}

The equivalent transfer function of the control law in \eqref{eq:DistribControlLaw} is
\begin{equation}\label{eq:PIDDControllerTransferFuntion}
	C(s) = \frac{U_{i}(s)}{\mathcal{E}_{i}(s)} = \frac{\kappa_{s}}{s} + \kappa_{p} + \kappa_{v} s + \kappa_{a} s^{2}
\end{equation}
and from the longitudinal state-space model in \eqref{eq:platoonDynamicsTransition}, the transfer function from the disturbance $\psi_{i}(t)$ to the output $y_{i}(t)$ for $\epsilon_{i}(t) = 0$ is:
\begin{equation}\label{eq:LongModelTransferFunction}
	G(s) = \frac{Y_{\! i}(s)}{\Psi_{\! i}(s)} = \mathbf{C} (s \mathbf{I}_{3} - \mathbf{A})^{-1} \mathbf{B} = \frac{1}{s^{2} \left( \varsigma s + 1 \right)}
\end{equation}
Then the tracking spacing error between $\psi_{i}(t)$ and $y_{i}(t)$ is
\begin{equation}
	E_{i}(s) = - \frac{G(s)}{1 + C(s) G(s)} \Psi_{i}(s) \label{eq:TrackingSpacingErrorDisturbance}
\end{equation}
For a step disturbance at the process input with gain $\kappa_{\psi}^{i} \neq 0$, the null steady-state spacing error is obtained only if $\kappa_{s} \neq 0$, see \eqref{eq:ErrorInputSteadyState}. 

This is also the case for parametric uncertainties, represented by $\Delta G$ in Fig.~\ref{fig:blockdiagram}. A detailed discussion about this aspect is given in \cite{Skogestad_2005_Book}. Finally, robustness to disturbances at the process output, represented by $\bar{\psi}_{i}$ in Fig.~\ref{fig:blockdiagram}, is ensured by similar arguments. 

\begin{figure*}[!t]
	\normalsize   
	\begin{equation}\label{eq:ErrorInputSteadyState}
		\varepsilon_{\text{ss}}^{i} = \lim_{t \rightarrow \infty} \varepsilon_{i}(t) = \lim_{s \rightarrow 0} s E_{i}(s) = - \lim_{s \rightarrow 0} s \frac{1}{\varsigma s^{3} + \left( 1 + \kappa_{a} \right) s^{2} + \kappa_{v} s + \kappa_{p} + \nicefrac{\kappa_{s}}{s}} \frac{\kappa_{\psi}^{i}}{s} = \begin{cases}
			-\nicefrac{\kappa_{\psi}^{i}}{\kappa_{p}}, & \kappa_{s} = 0, \\
			0, & \kappa_{s} \neq 0.
		\end{cases}
	\end{equation}	
	\hrulefill    
	\vspace*{4pt} 
\end{figure*}

\end{remark}

\begin{remark}[Restrictions on $\kappa_{p}$]\label{rem:KappaPConstraints}

In \eqref{eq:ControlGainsLookAhead} and \eqref{eq:ControlGainsUndirected}, one can notice that $\kappa_{p}$ depends on $\kappa_{v}$ and vice-versa. Additionally, the expression for $\kappa_{p}$ is a special case of that obtained for $\kappa_{v}$ when $\kappa_{s} = 0$, for a particular $\lambda_{i}$. Thus, the restrictions on it reduces to $\kappa_{p} > 0$. 

\end{remark}

%
%
%
%
%

%
%
%

\section{Simulation Results and Analysis}\label{sec:ResultsAnalysis}
	
In this section, we carried out experiments using \emph{CARLA} (\textit{Car Learning to Act}), a well-known open-source simulator for driving research that provides a realistic environment for testing and developing algorithms for autonomous vehicles \cite{Dosovitskiy_2017_InProceedings}.
We start by showing that bounded step disturbances (like road slope, wind gusts, and parametric modeling uncertainties) lead to formation errors in the platoon. Then, we demonstrate how the proposed control law can address this problem with properly designed gains.
%

Our platoon, illustrated in Fig.~\ref{fig:carla_simulator}, consists of ten identical vehicles, one leader, and nine followers, whose parameters are shown in Tab.~\ref{tab:VehicleParameters}.
The world file employed in the tests was the \emph{Town06}, a \emph{CARLA} scenario with a long straight road. To avoid modifying it, the road slope and wind disturbances were emulated by applying external forces to the vehicles' center of mass.

\begin{figure}[t]
    \centering
    \includegraphics[width=\linewidth]{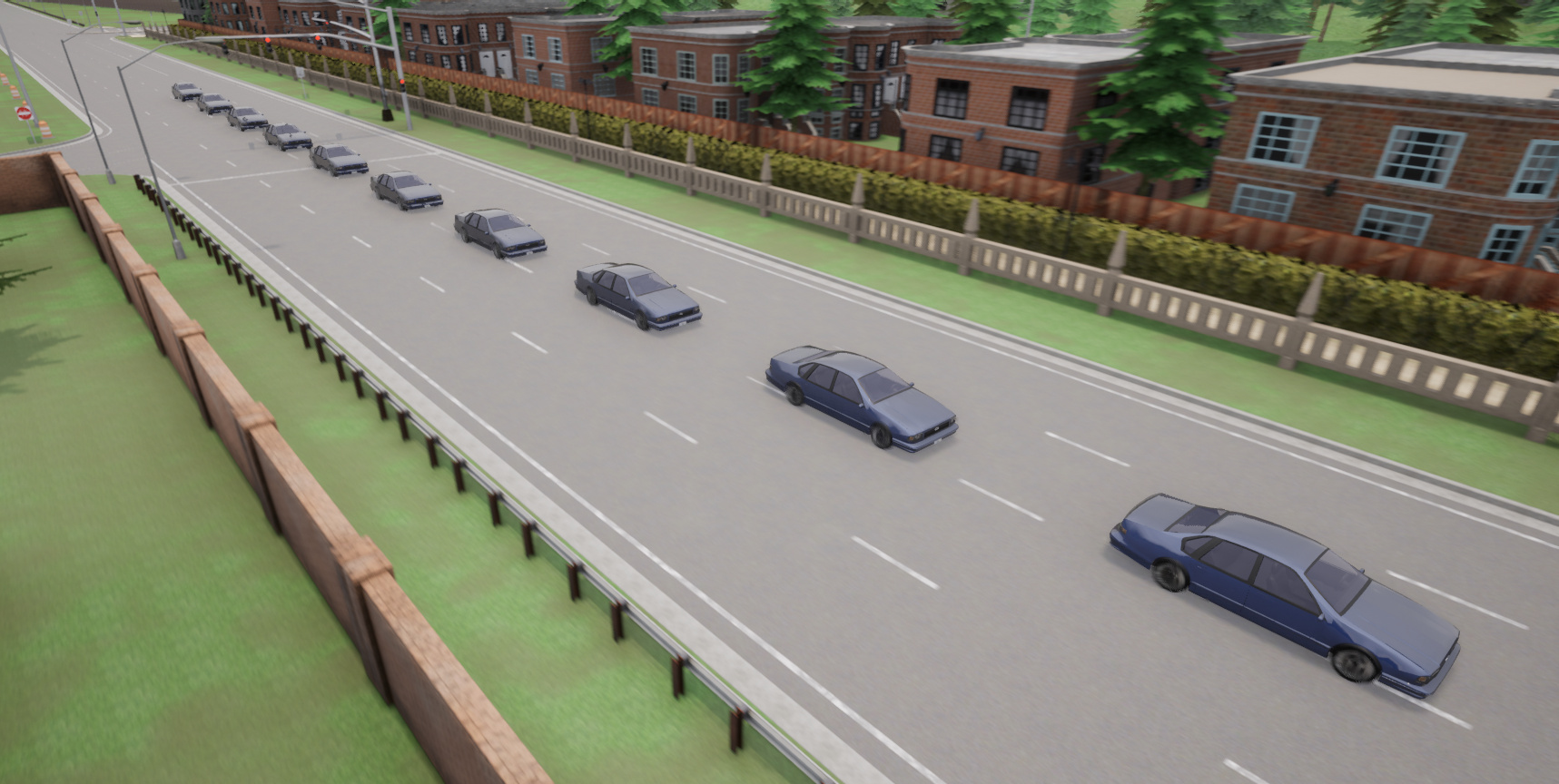}
        \caption{Homogeneous vehicular platoon in the \emph{CARLA} simulator with one leader and nine followers subjected to uncertainties in the parameters, road slope, and wind disturbances.}
    \label{fig:carla_simulator}
\end{figure}

\begin{table}[htb]
	\centering
	\caption{Vehicle's parameters estimated from \emph{CARLA}.}
	\label{tab:VehicleParameters}
	\begin{tabular}{c|c|c|c}
		\toprule \multicolumn{4}{c}{Parameter values} \\ \midrule
		$m = \SI{1613}{\kg}$ 	& $\eta = 1.0$	& $r = \SI{0.34}{\m}$ 				& $g = \SI{9.8}{\m/\s^{2}}$ \\
		$\varsigma = 0.15$ 		& $c_{d} = 0.62$ & $\rho = \SI{1.225}{\kg/\m^{3}}$ 	& $\mu = 0.01$ \\ \bottomrule
	\end{tabular}
\end{table}

Here, the parametric uncertainties result from a difference between some parameters in Table \ref{tab:VehicleParameters} and the values used internally by \emph{CARLA}. For example, the simulator employs three different $\varsigma$ values for different operating conditions, but in the cancellation process, we only use one. We also do not have direct access to the values of $\eta$ and $\mu$, then we used arbitrarily estimated values. Meanwhile, $c_{d}$ was approximated by the formula $c_{d} = c_k \big[1.6 + 0.00056 (m - 765) \big]$, where $c_k$ is the drag coefficient provided by \emph{CARLA}. Due to the lack of more precise information about all these parameters, we have an imperfect torque linearization in Eq.~\eqref{eq:FeedbackLinearization}.

For all trials, the agents start in $p_{i}(0) = - i \cdot d_{ij}$ (with $d_{ij} = \SI{10}{\m}$ being the reference spacing distance), initial speed $v_{i}(0) = \SI{15}{\m/\s}$, and null acceleration $a_{i}(0) = \SI{0}{\m/\s}$. 
The platoon is further disturbed as the leader accelerates 
\[
    u_{0}(t) = 
    \begin{cases}
	\SI{1}{\nicefrac{\m}{\s^2}}, & 30 \leqslant t \leqslant 35 \\
	\SI{0}{\nicefrac{\m}{\s^2}}, & \text{otherwise}
    \end{cases}.
\]
This is followed by a step of $\SI{10}{\degree}$ in the road slope, subsequently applied to each follower as they cross a specified position (reached after the leader travels for 100 seconds), 
\[
    \theta_{0}(t) = 
   \begin{cases}
        \SI{10}{\degree}, & p_{0}(t) \geqslant \SI{1680}{\m}\\
        \SI{0}{\degree}, & \text{otherwise}
    \end{cases}.
\]

Finally, a similar procedure is performed for the wind gust disturbance, with a step of $-\SI{20}{\m/\s}$, reaching each platoon's agent as they travel,
\[
    v_{w,0}(t) = 
    \begin{cases}
        \SI{20}{\nicefrac{\m}{\s}}, & t \geqslant \SI{150}{\s} \\
        \SI{0}{\nicefrac{\m}{\s}}, & \text{otherwise}
    \end{cases}.
\]


To satisfy Theorems \ref{thm:LookAheadStability}-\ref{thm:BidirectionalStability} and Corollaries \ref{cor:LookAheadKappaSZero}-\ref{cor:BidirectionalKappaSZero}, we used the controller gains given in Tab.~\ref{table:ControllerGains}, and the simulations were conducted on the network topologies of Figs.~\ref{fig:Topologies1}-\ref{fig:Topologies2}.

\begin{table}[htb]
	\centering
	\caption{Control gains for each network topology.}
	\label{table:ControllerGains}
        \scalebox{1}{\begin{tabular}{c|cccc|ccc}
			\toprule
			& \multicolumn{4}{c|}{ Theorems \ref{thm:LookAheadStability}-\ref{thm:BidirectionalStability} } & \multicolumn{3}{c}{ Corollaries \ref{cor:LookAheadKappaSZero}-\ref{cor:BidirectionalKappaSZero} } \\
			\midrule
            Topology & $\kappa_{s}$ & $\kappa_{p}$ & $\kappa_{v}$ & $\kappa_{a}$ & $\kappa_{p}$ & $\kappa_{v}$ & $\kappa_{a}$ \\
			\midrule
			PF & 0.150 & 1.0 & 3.450 & 1.000 &1.0 &2.150 & 1.000 \\
			PFL & 0.075 & 1.0 & 3.225 & 1.500 &1.0 &2.075 & 1.500 \\
			TPF & 0.075 & 1.0 & 3.225 & 1.500 &1.0 &2.075 & 1.500 \\
			TPFL & 0.050 & 1.0 & 3.150 & 1.667 &1.0 &2.050 & 1.667 \\
			$r$PF & 0.030 & 1.0 & 3.090 & 1.800 &1.0 &2.030 & 1.800 \\
			$r$PFL & 0.025 & 1.0 & 3.075 & 1.833 &1.0 &2.025 & 1.833 \\ \midrule
			BD & 0.010 & 1.0 & 5.086 & 1.743 &1.0 &2.286 & 1.743 \\
			BDL & 0.010 & 1.0 & 1.052 & 1.795 &1.0 &2.107 & 1.795 \\
			$r$BD & 0.010 & 1.0 & 1.423 & 1.890 &1.0 &2.175 & 1.890 \\
			$r$BDL & 0.010 & 1.0 & 1.103 & 1.900 &1.0 &2.103 & 1.900 \\
			\bottomrule
	\end{tabular}}
\end{table}


Fig.~\ref{fig:SpacingErrorTopologiesCarla} presents the spacing error between consecutive vehicles for the proposed control law \eqref{eq:DistribControlLaw}. In Figures \ref{fig:SpacingErrorTopologiesCarla} (a), (c), (e), (g), (i), (k), (m), (o), (q), and (s) we demonstrate the effects of the disturbances in all evaluated topologies, while in Figures \ref{fig:SpacingErrorTopologiesCarla} (b), (d), (f), (h), (j), (l), (n), (p), (r) and (t) we used a term $\kappa_{s} \neq 0$ to reach null steady-state errors.

As one can easily see in all experiments (in some more than others), there is a small initial error caused by the imperfect cancellation of the system's nonlinearities, an effect that becomes evident after the leader accelerates around the first 30 seconds. Here, it is possible to say that our adequate parameter estimates (Tab.~\ref{tab:VehicleParameters}) were sufficient to mitigate these effects, but even so, in the experiments with the control integral term, this issue was completely eliminated.

Still analyzing the experiments in general, it is noticeable that the most deleterious effect on the platooning control was the road slope, causing significant spacing errors between agents. In this situation, the contributions of the proposed integral control action are even more evident, since it is capable of rejecting the disturbance in all experiments. 
At this point, it is important to highlight that $\SI{10}{\degree}$, about $18\%$ of road slope, is a significant value. Just as a reference, the AASHTO (\textit{American Association of State Highway and Transportation Officials}) establishes guidelines for the longitudinal slope of lanes, which usually vary from $3\%$ to $6\%$, depending on several factors, such as the category of the highway and the speed of travel \cite{Hancock_2013_Article}. 
Finally, a smaller but also significant effect is caused by the wind effect, duly rejected by our proposed controller.

Now considering only topologies in which all agents have direct communication with the leader, it is possible to observe that the first follower always has a non-null spacing error, while others have almost null spacing errors when $k_{s} = 0$. According to \cite{Zheng_2016_Article}, the reason is that all the followers present a similar behavior due to their homogeneous dynamics and given that they have information about the state of the leader, mitigating the spacing errors among the followers. In any case, the proposed integral action is capable of canceling out steady-state errors caused by constant disturbances acting on all agents in the platoon.

Furthermore, there were no collisions between vehicles within the platoon. This observation holds true even for the BD topology, as depicted in Fig. \ref{fig:SpacingErrorTopologiesCarla} (n), where a minimum error of $-\SI{5}{\m}$ is recorded for a desired spacing distance of $\SI{10}{\m}$. We also note that the proposed control law ensures effective disturbance rejection for a rigid spacing policy, which indicates its potential employment on more flexible policies, including constant and variable time headway spacing strategies.

\begin{figure*}
	\centering
	\subfloat[PF  with $k_{s} = 0$.]{\includegraphics[width=0.243\textwidth]{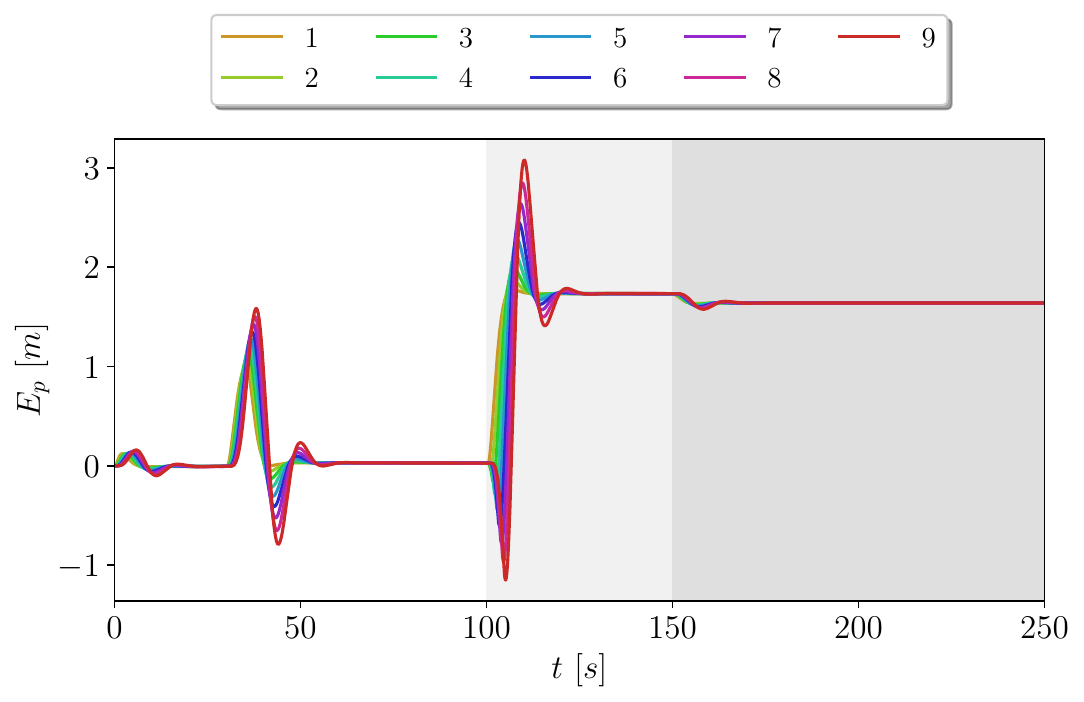} \label{subfig:PF_nointeg}}
	\subfloat[PF  with $k_{s} > 0$.]{\includegraphics[width=0.243\textwidth]{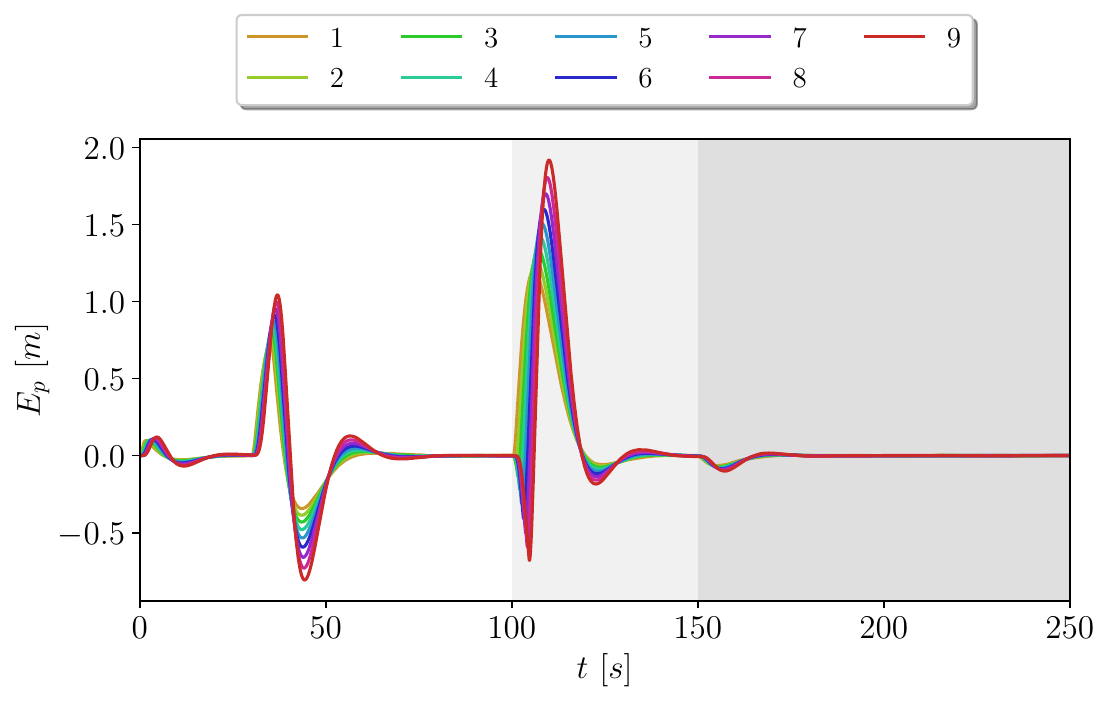} \label{subfig:PF_integ}}
	\subfloat[PFL  with $k_{s} = 0$.]{\includegraphics[width=0.243\textwidth]{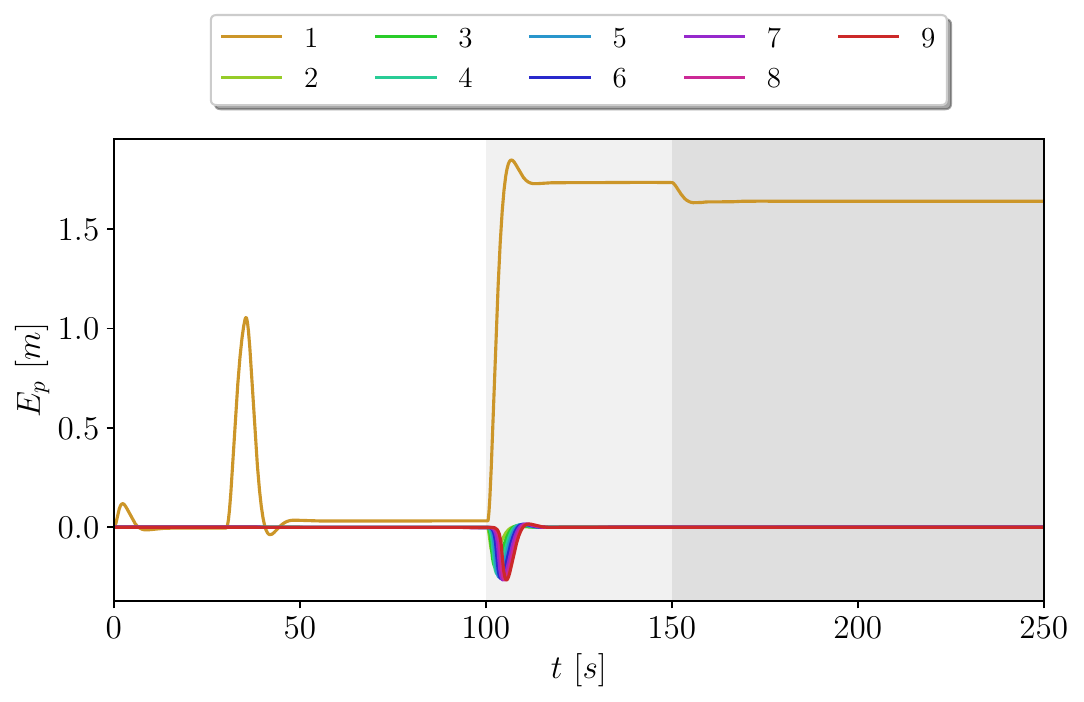} \label{subfig:PFL_nointeg}}  
	\subfloat[PFL  with $k_{s} > 0$.]{\includegraphics[width=0.243\textwidth]{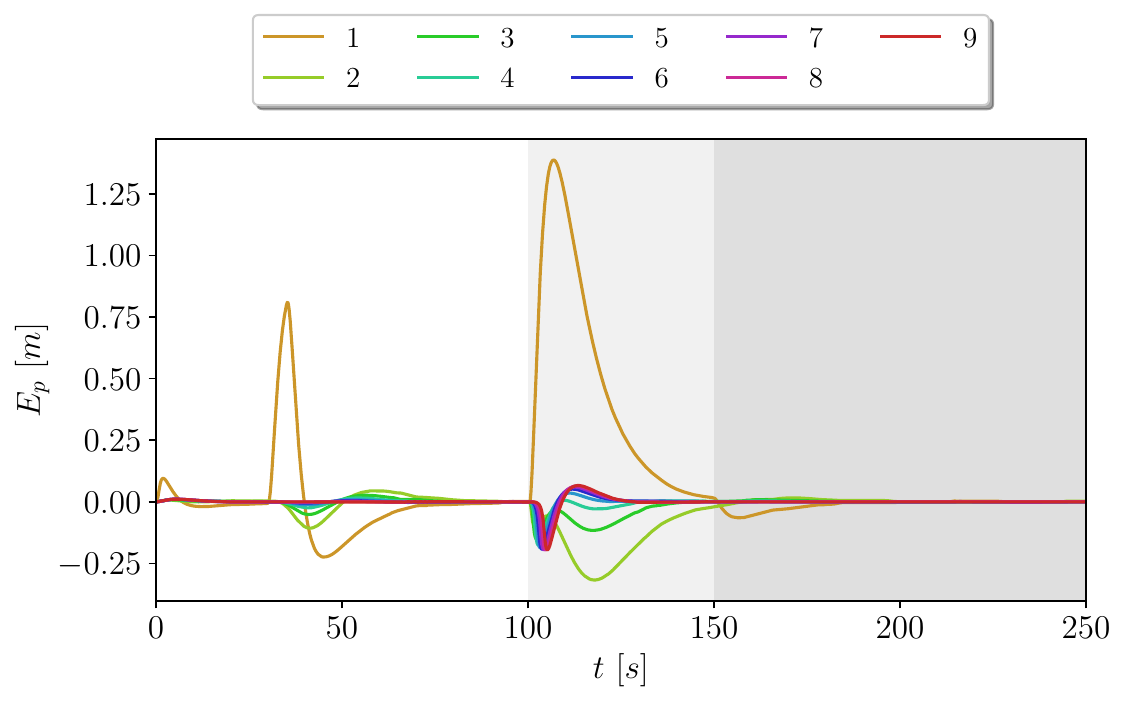} \label{subfig:PFL_integ}} \\
	\subfloat[TPF  with $k_{s} = 0$.]{\includegraphics[width=0.243\textwidth]{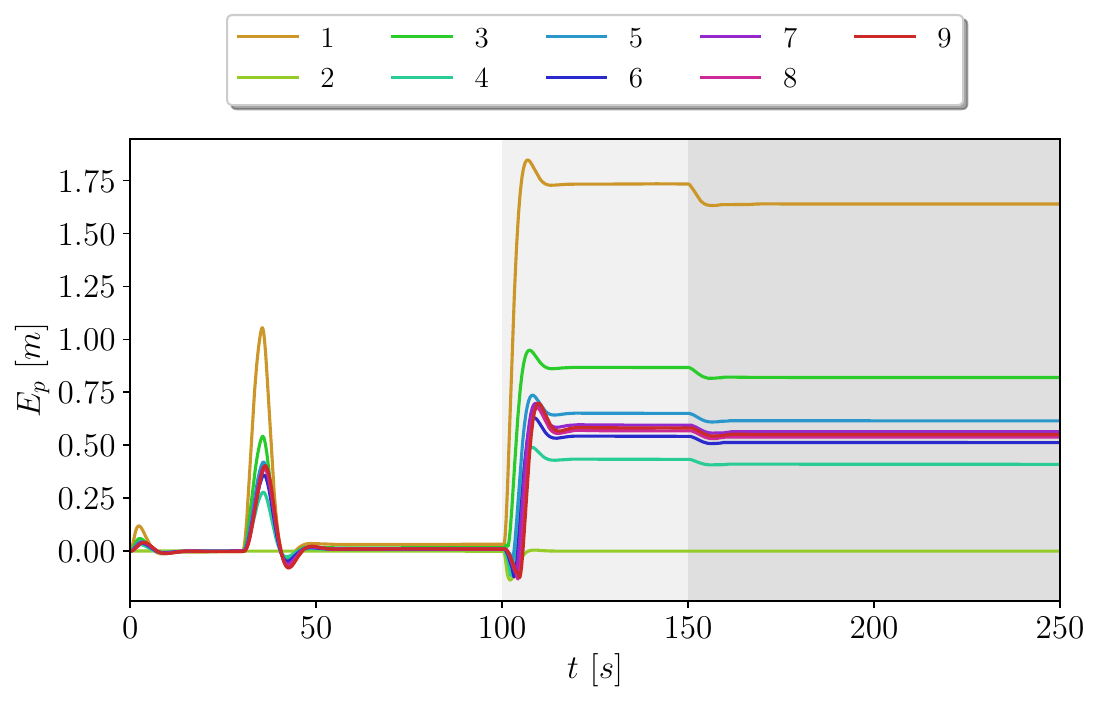} \label{subfig:TPF_nointeg}}
	\subfloat[TPF  with $k_{s} > 0$.]{\includegraphics[width=0.243\textwidth]{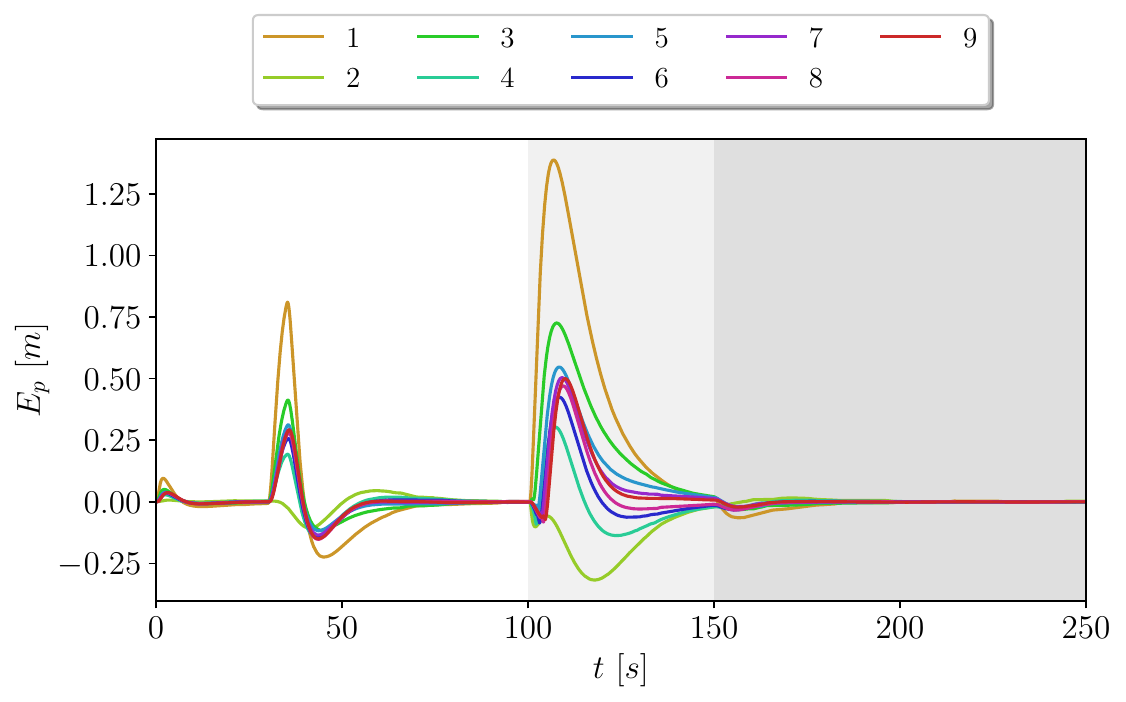} \label{subfig:TPF_integ}}
	\subfloat[TPFL  with $k_{s} = 0$.]{\includegraphics[width=0.243\textwidth]{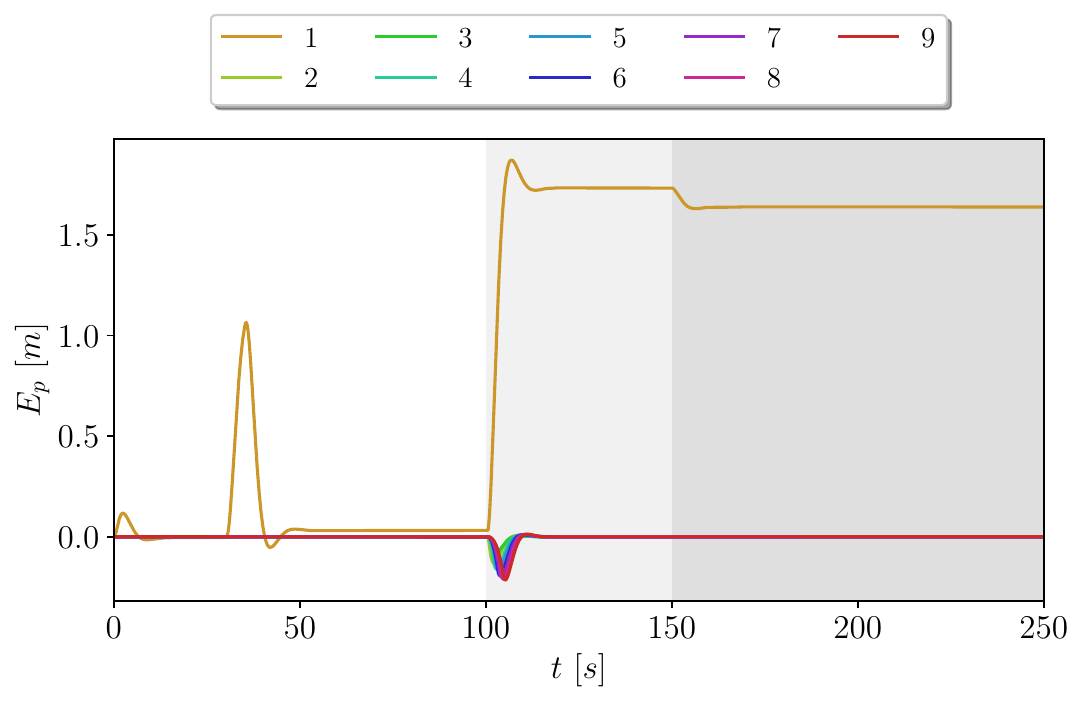} \label{subfig:TPFL_nointeg}}
	\subfloat[TPFL  with $k_{s} > 0$.]{\includegraphics[width=0.243\textwidth]{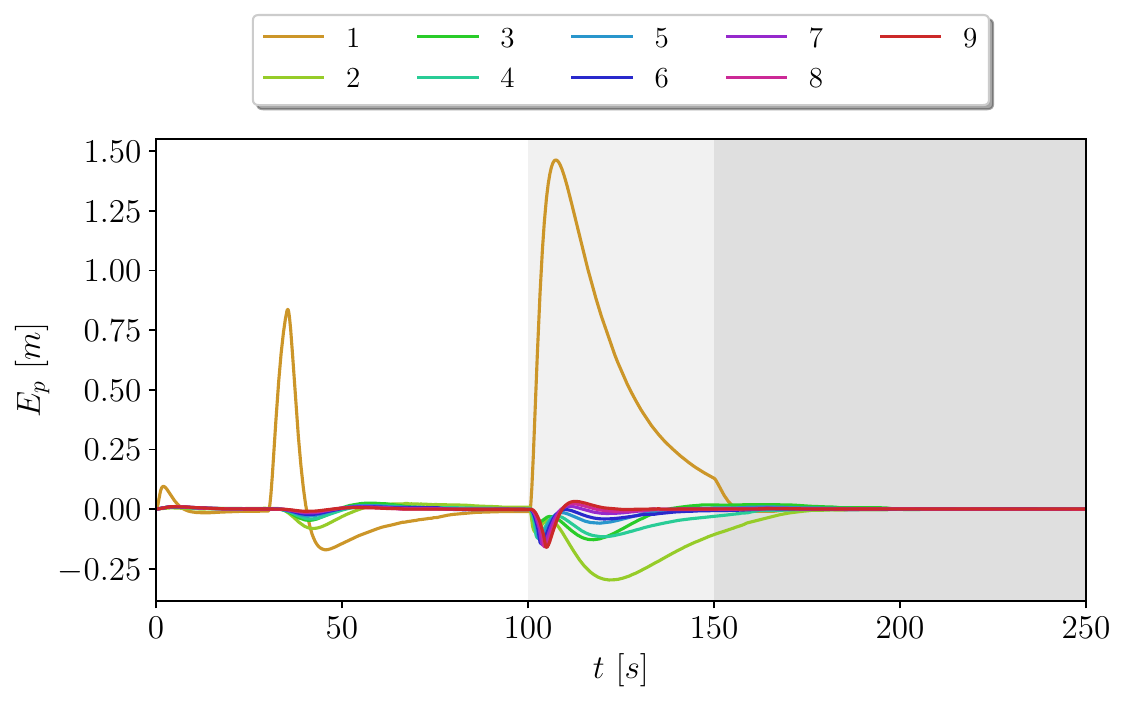} \label{subfig:TPFL_integ}} \\
	\subfloat[$r$PF  with $r = 5$ and $k_{s} = 0$.]{\includegraphics[width=0.243\textwidth]{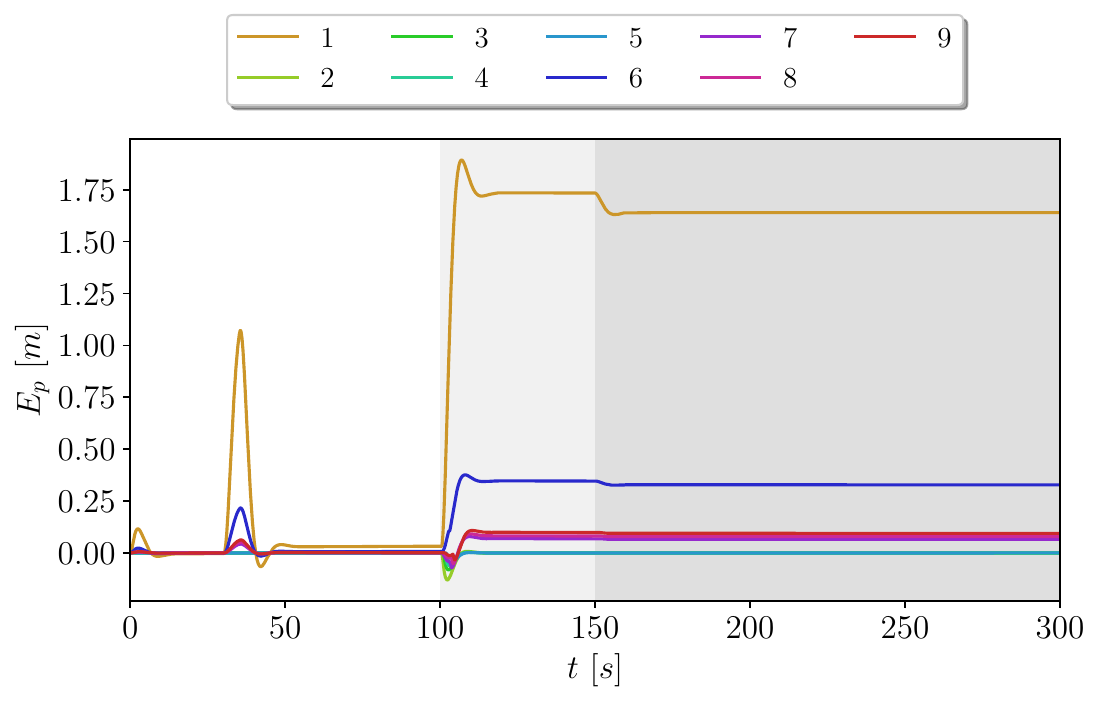} \label{subfig:rPF_nointeg}}
	\subfloat[$r$PF  with $r = 5$ and $k_{s} > 0$.]{\includegraphics[width=0.243\textwidth]{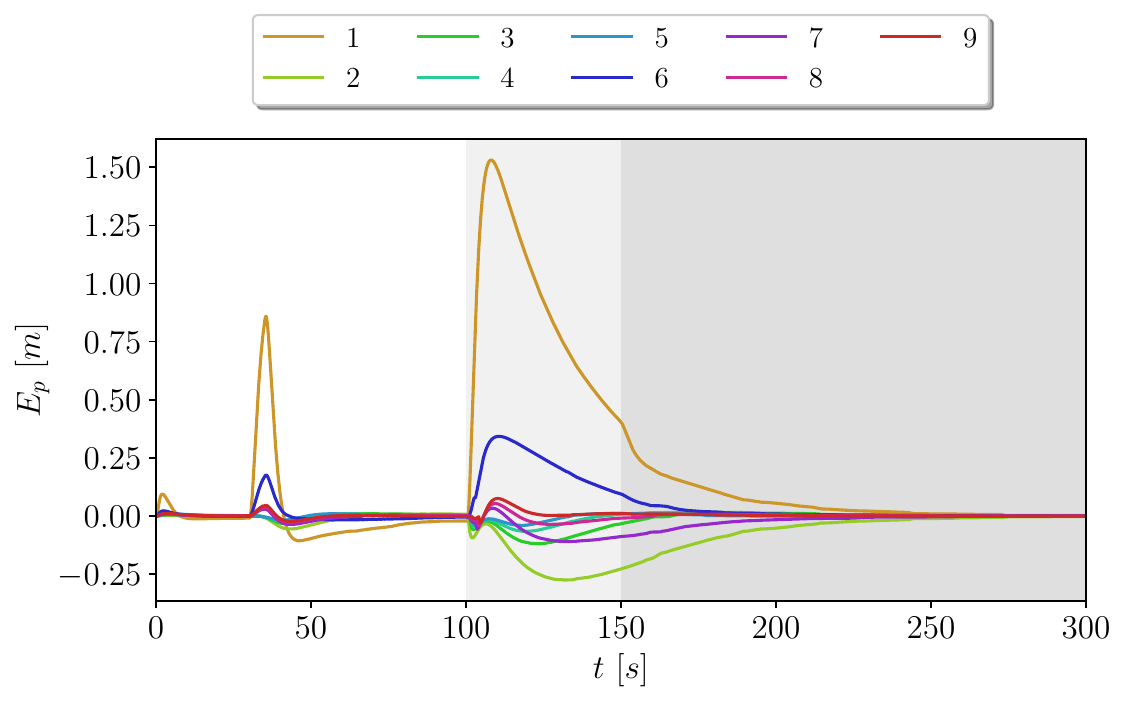} \label{subfig:rPF_integ}}
	\subfloat[$r$PFL  with $r = 5$ and $k_{s} = 0$.]{\includegraphics[width=0.243\textwidth]{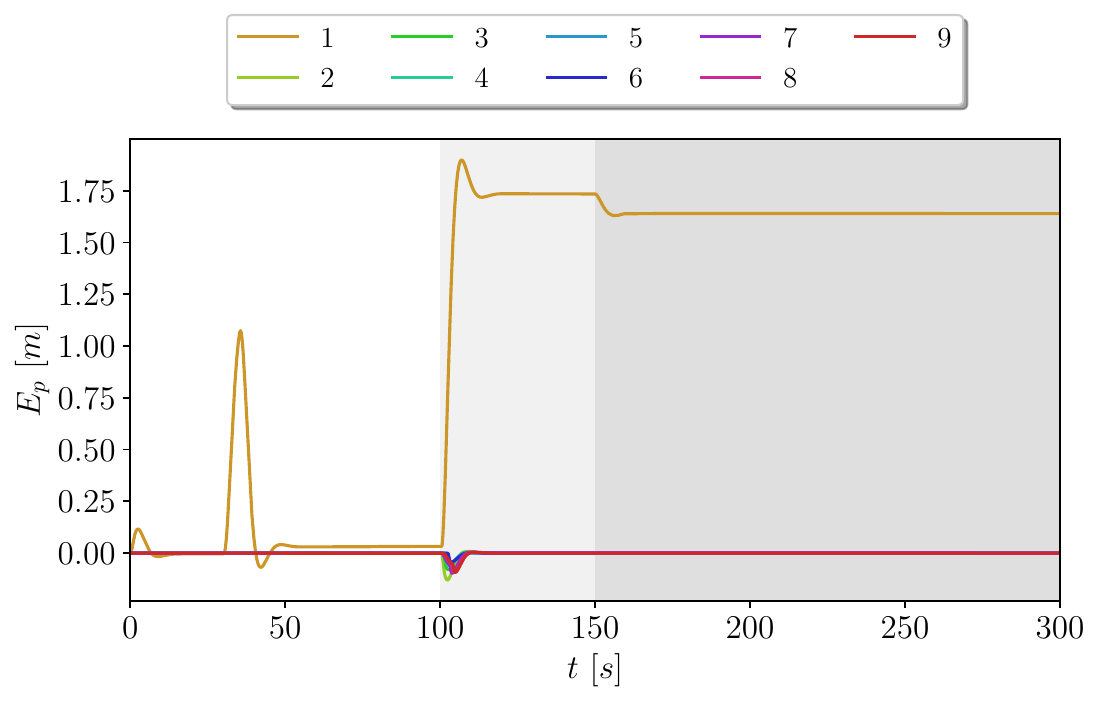} \label{subfig:rPFL_nointeg}}
	\subfloat[$r$PFL  with $r = 5$ and $k_{s} > 0$.]{\includegraphics[width=0.243\textwidth]{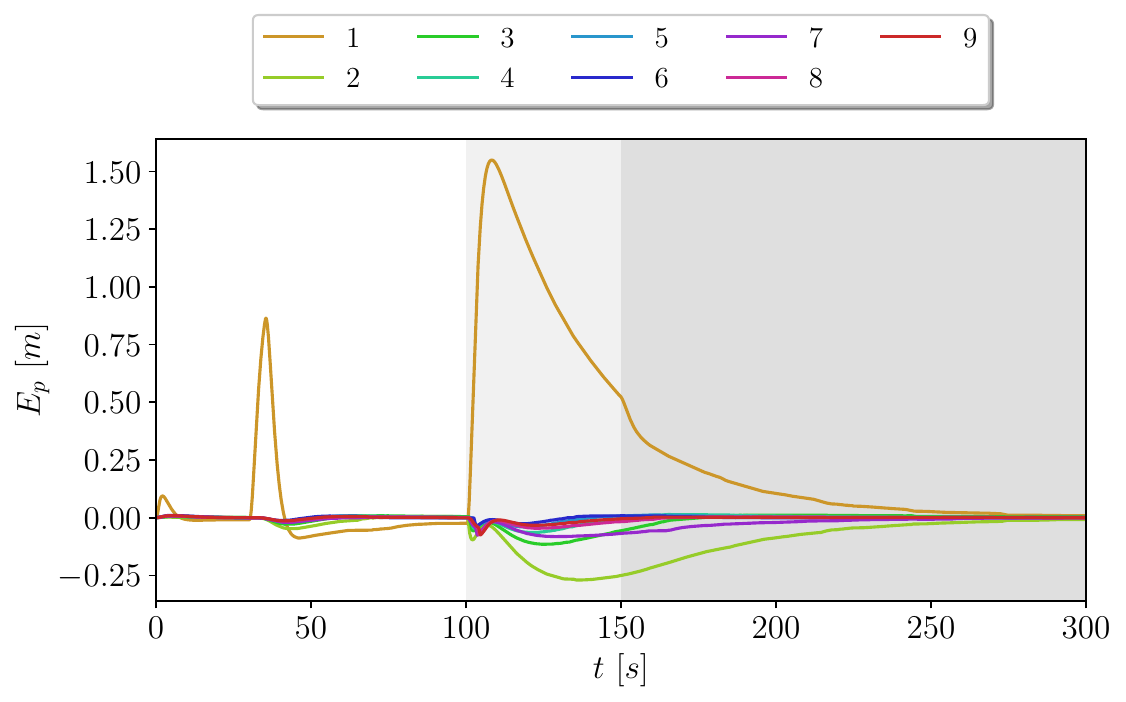} \label{subfig:rPFL_integ}} \\
	\subfloat[BD  with $k_{s} = 0$.]{\includegraphics[width=0.243\textwidth]{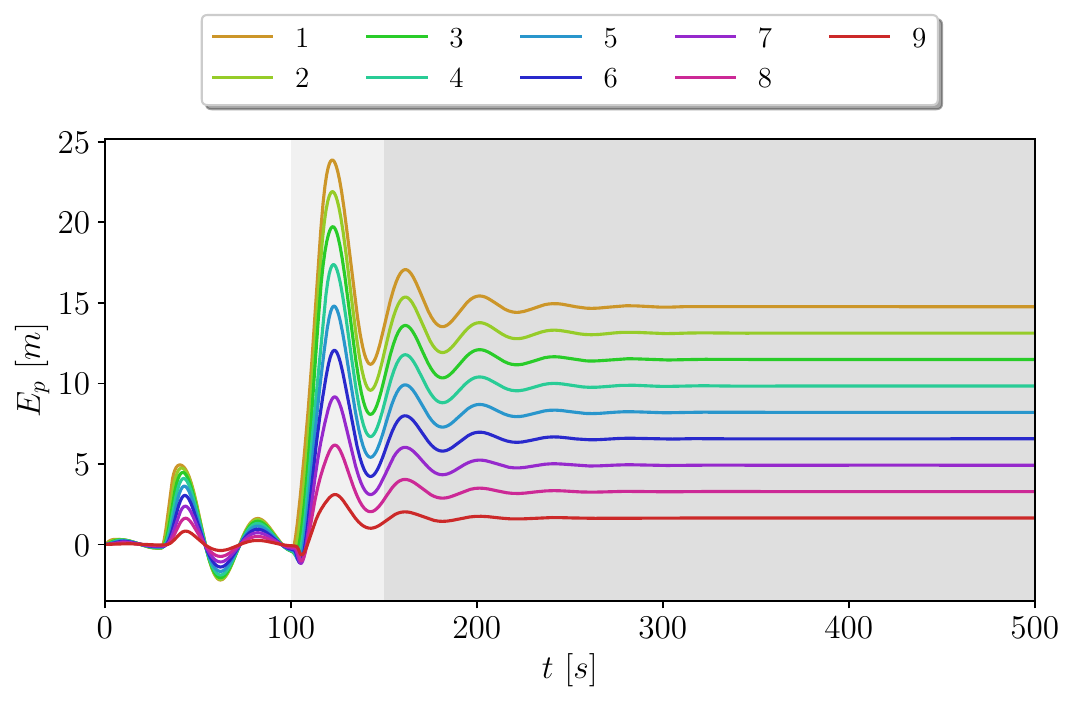} \label{subfig:BD_nointeg}}
	\subfloat[BD  with $k_{s} > 0$.]{\includegraphics[width=0.243\textwidth]{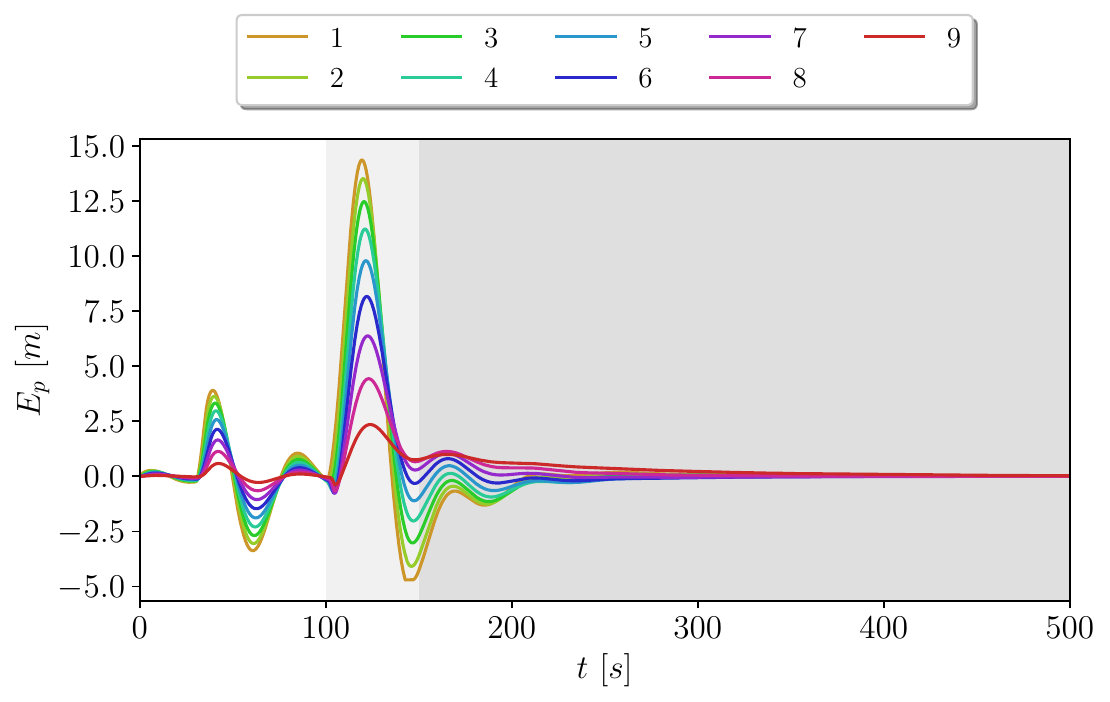} \label{subfig:BD_integ}}
	\subfloat[BDL  with $k_{s} = 0$.]{\includegraphics[width=0.243\textwidth]{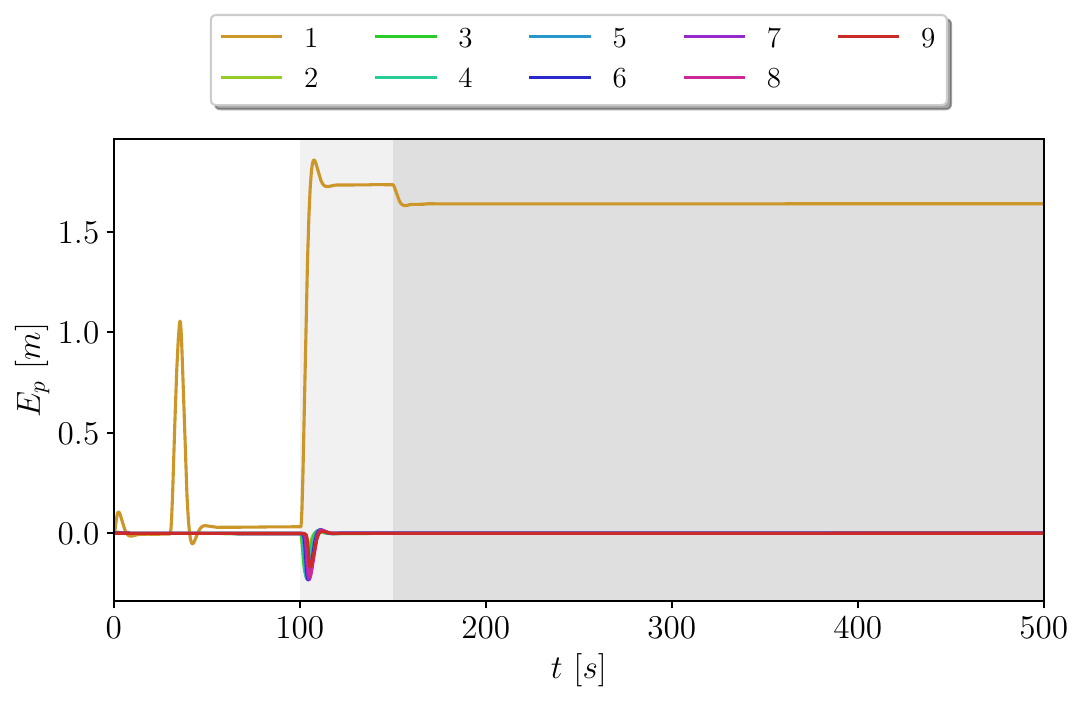} \label{subfig:BDL_nointeg}}  
	\subfloat[BDL  with $k_{s} > 0$.]{\includegraphics[width=0.243\textwidth]{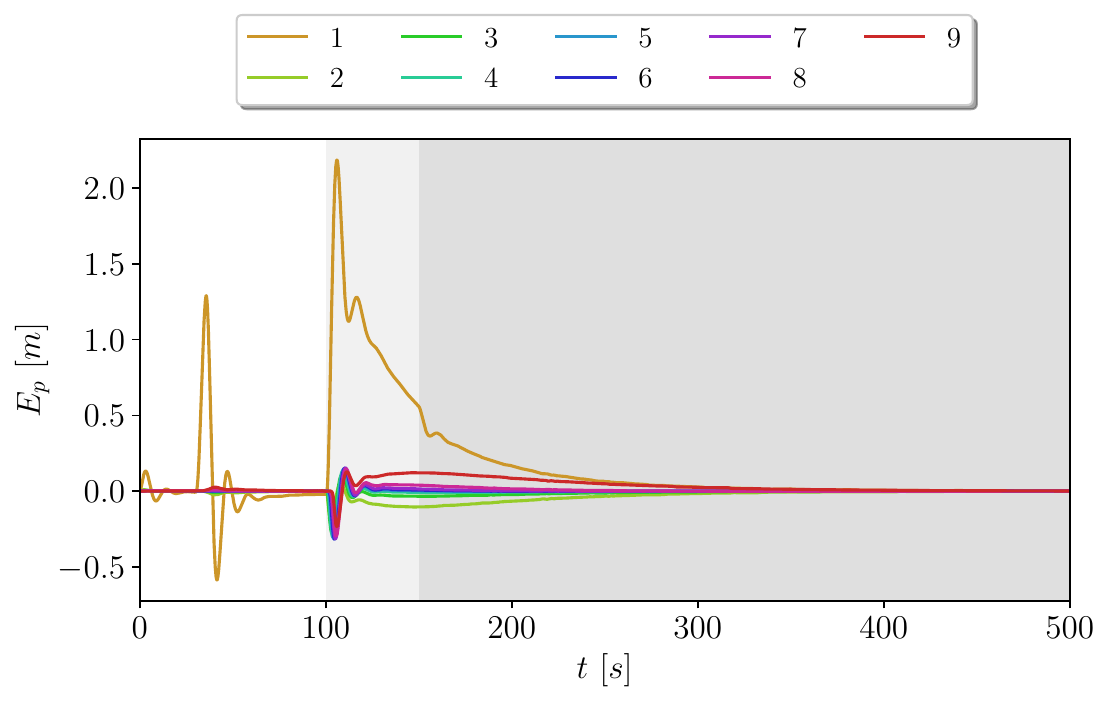} \label{subfig:BDL_integ}} \\
	\subfloat[$r$BD  with $k_{s} = 0$.]{\includegraphics[width=0.243\textwidth]{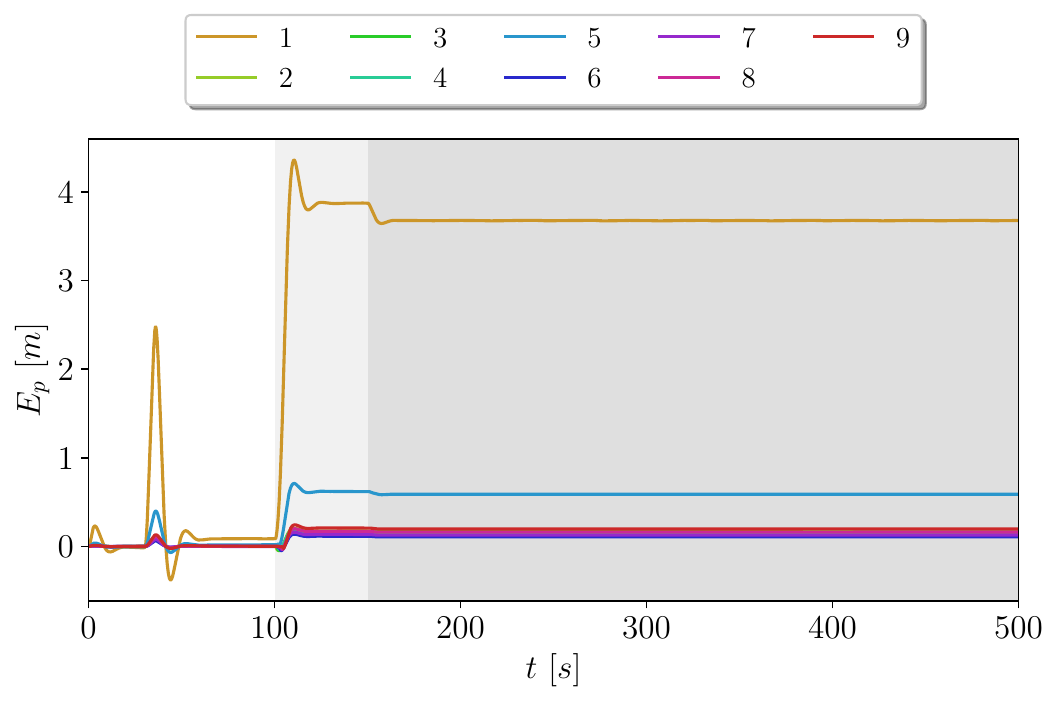} \label{subfig:rBD_nointeg}}
	\subfloat[$r$BD  with $k_{s} > 0$.]{\includegraphics[width=0.243\textwidth]{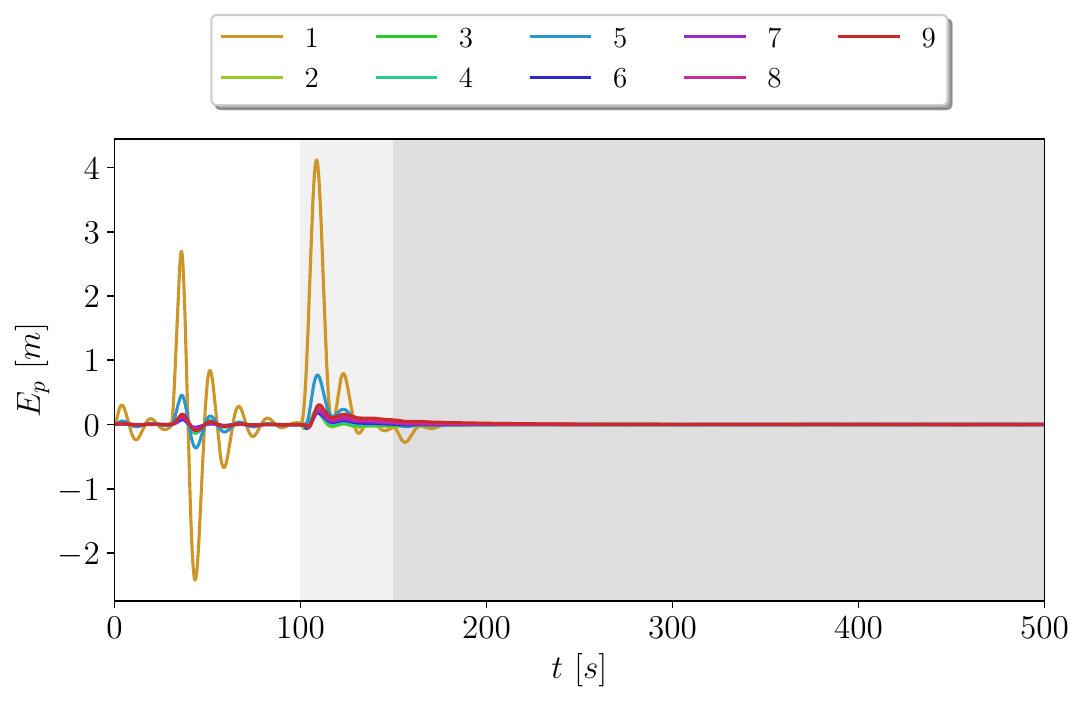} \label{subfig:rBD_integ}}
	\subfloat[$r$BDL  with $r = 4$ and $k_{s} = 0$.]{\includegraphics[width=0.243\textwidth]{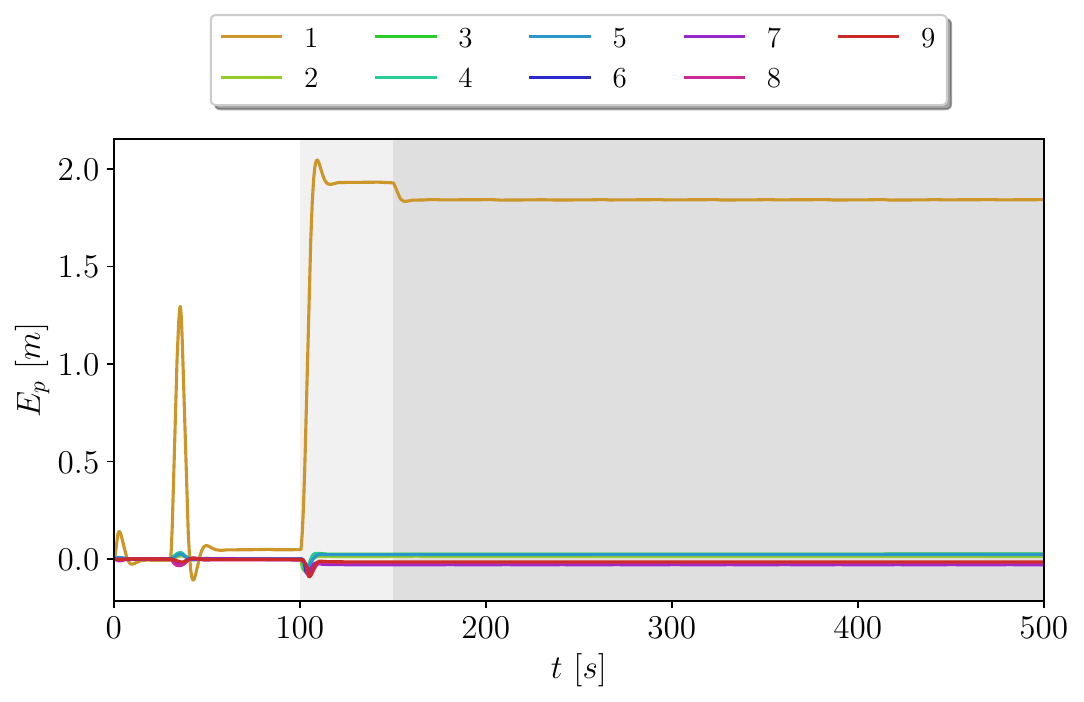} \label{subfig:rBDL_nointeg}}
	\subfloat[$r$BDL  with $r = 4$ and $k_{s} > 0$.]{\includegraphics[width=0.243\textwidth]{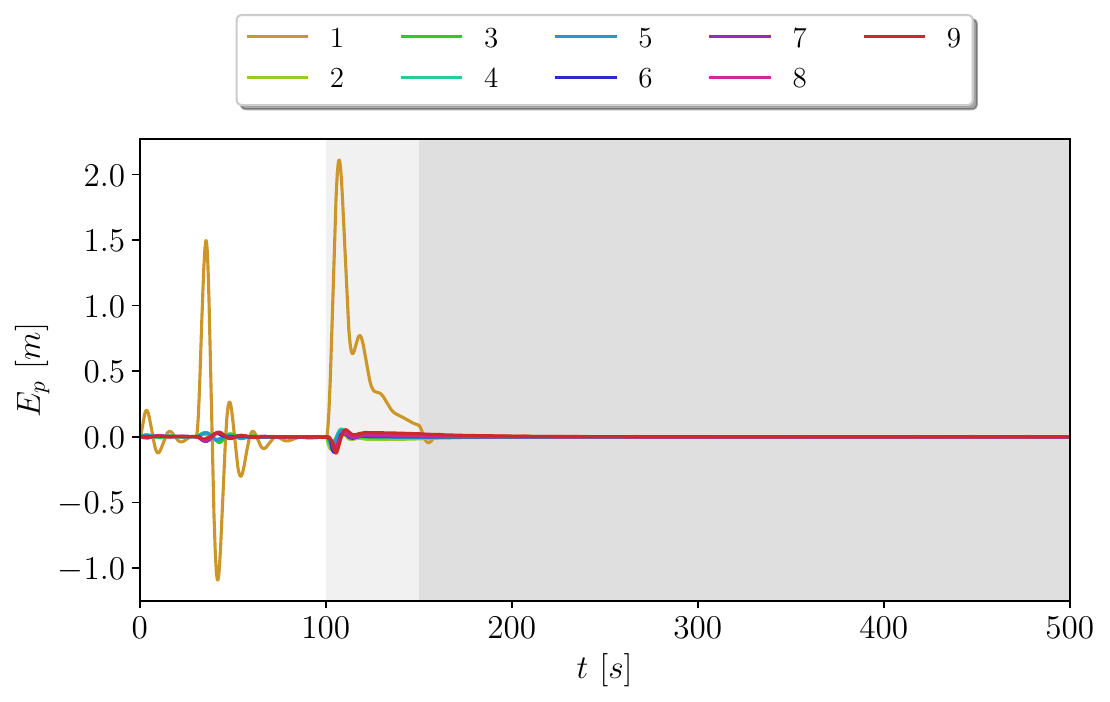} \label{subfig:rBDL_integ}}
	 \caption{Spacing error profiles for a platoon connected through different network topologies and subjected to parametric model uncertainty, leading vehicle acceleration (white background), road slope (light gray background), and wind gust (dark gray background) disturbances.}
	\label{fig:SpacingErrorTopologiesCarla}
\end{figure*}

\section{Conclusion and Future Work}\label{sec:Conclusion}
	
In this paper, we have investigated the longitudinal control of a platoon of vehicles subjected to internal and external disturbances. We focus on homogeneous platoons with a constant spacing policy under generalized communication topologies. We demonstrated that a simple change in the linear distributed control law can effectively reject constant bounded disturbances, such as those caused by modeling uncertainties, road slopes, and wind gusts, independently of the network topology. Our proposal was validated using the \emph{CARLA}, a highly realistic simulator of autonomous vehicles widely used by researchers in the field of intelligent transportation.

Aspects for further investigation include extending the modeling to heterogeneous platoons and imperfect communication channels with time delays and denial of service. We can also address more challenging problems related to limited communication range, time-varying network topologies, and the influence of mixed traffic on the platoon's stability.




\bibliography{references_hom}

\begin{thebibliography}{10}
\providecommand{\url}[1]{#1}
\csname url@samestyle\endcsname
\providecommand{\newblock}{\relax}
\providecommand{\bibinfo}[2]{#2}
\providecommand{\BIBentrySTDinterwordspacing}{\spaceskip=0pt\relax}
\providecommand{\BIBentryALTinterwordstretchfactor}{4}
\providecommand{\BIBentryALTinterwordspacing}{\spaceskip=\fontdimen2\font plus
\BIBentryALTinterwordstretchfactor\fontdimen3\font minus \fontdimen4\font\relax}
\providecommand{\BIBforeignlanguage}[2]{{%
\expandafter\ifx\csname l@#1\endcsname\relax
\typeout{** WARNING: IEEEtran.bst: No hyphenation pattern has been}%
\typeout{** loaded for the language `#1'. Using the pattern for}%
\typeout{** the default language instead.}%
\else
\language=\csname l@#1\endcsname
\fi
#2}}
\providecommand{\BIBdecl}{\relax}
\BIBdecl

\bibitem{Gao_2016_Article}
F.~Gao, S.~E. Li, Y.~Zheng, and D.~Kum, ``Robust control of heterogeneous vehicular platoon with uncertain dynamics and communication delay,'' \emph{{IET} Intelligent Transport Systems}, vol.~10, no.~7, pp. 503--513, sep 2016.

\bibitem{Boriboonsomsin_TRR_2009}
K.~Boriboonsomsin and M.~Barth, ``Impacts of road grade on fuel consumption and carbon dioxide emissions evidenced by use of advanced navigation systems,'' \emph{Transportation Research Record}, vol. 2139, no.~1, pp. 21--30, 2009.

\bibitem{Ciuffo_TRC_2021}
B.~Ciuffo, K.~Mattas, M.~Makridis, G.~Albano, A.~Anesiadou, Y.~He, S.~Josvai, D.~Komnos, M.~Pataki, S.~Vass \emph{et~al.}, ``Requiem on the positive effects of commercial adaptive cruise control on motorway traffic and recommendations for future automated driving systems,'' \emph{Transportation Research Part C: Emerging Technologies}, vol. 130, p. 103305, 2021.

\bibitem{Zhang_TITS_2020}
J.~Zhang and H.~Jin, ``Optimized calculation of the economic speed profile for slope driving: Based on iterative dynamic programming,'' \emph{IEEE Trans. Intell. Transp. Syst.}, vol.~23, no.~4, pp. 3313--3323, 2020.

\bibitem{Seiler_2004_Article}
P.~Seiler, A.~Pant, and K.~Hedrick, ``Disturbance {P}ropagation in {V}ehicle {S}trings,'' \emph{{IEEE} Transactions on Automatic Control}, vol.~49, no.~10, pp. 1835--1841, oct 2004.

\bibitem{Zheng_2016_Article}
Y.~Zheng, S.~E. Li, J.~Wang, D.~Cao, and K.~Li, ``{S}tability and {S}calability of {H}omogeneous {V}ehicular {P}latoon: {S}tudy on the {I}nfluence of {I}nformation {F}low {T}opologies,'' \emph{IEEE Trans. Intell. Transp. Syst.}, no.~1, pp. 14--26, 2016.

\bibitem{Zheng_2018_Article}
Y.~Zheng, S.~E. Li, K.~Li, and W.~Ren, ``Platooning of {C}onnected {V}ehicles {W}ith {U}ndirected {T}opologies: {R}obustness {A}nalysis and {D}istributed {H}-infinity {C}ontroller {S}ynthesis,'' \emph{IEEE Trans. Intell. Transp. Syst.}, vol.~19, pp. 1353--1364, 2018.

\bibitem{Li_2022_Article}
Y.~Li, Q.~Lv, H.~Zhu, H.~Li, H.~Li, S.~Hu, S.~Yu, and Y.~Wang, ``Variable {T}ime {H}eadway {P}olicy {B}ased {P}latoon {C}ontrol for {H}eterogeneous {C}onnected {V}ehicles {W}ith {E}xternal {D}isturbances,'' \emph{IEEE Trans. Intell. Transp. Syst.}, vol.~23, no.~11, pp. 21\,190--21\,200, nov 2022.

\bibitem{Zhao_2022_Articlea}
Y.~Zhao, Z.~Liu, and W.~S. Wong, ``Resilient {P}latoon {C}ontrol of {V}ehicular {C}yber {P}hysical {S}ystems {U}nder {D}o{S} {A}ttacks and {M}ultiple {D}isturbances,'' \emph{IEEE Trans. Intell. Transp. Syst.}, vol.~23, no.~8, pp. 10\,945--10\,956, 2022.

\bibitem{Souza_2020_Article}
F.~O. Souza, L.~A.~B. Torres, L.~A. Mozelli, and A.~A. Neto, ``Stability and {F}ormation {E}rror of {H}omogeneous {V}ehicular {P}latoons with {C}ommunication {T}ime {D}elays,'' \emph{IEEE Trans. Intell. Transp. Syst.}, vol.~21, no.~10, pp. 4338--4349, oct 2020.

\bibitem{Neto_2019_Article}
A.~A. Neto, L.~A. Mozelli, and F.~O. Souza, ``Control of air-ground convoy subject to communication time delay,'' \emph{Computers {\&} Electrical Engineering}, vol.~76, pp. 213--224, jun 2019.

\bibitem{Godinho_2022_Article}
D.~A. Godinho, A.~A. Neto, L.~A. Mozelli, and F.~O. Souza, ``Control and {R}eorganization of {H}eterogeneous {V}ehicle {P}latoons after {V}ehicle {E}xits and {E}ntrances,'' \emph{International Journal of Control, Automation and Systems}, 2022.

\bibitem{Godinho_2023_Article}
------, ``A {S}trategy for {T}raffic {S}afety of {V}ehicular {P}latoons {U}nder {C}onnection {L}oss and {T}ime-{D}elay,'' \emph{IEEE Trans. Intell. Transp. Syst.}, vol.~24, no.~6, pp. 6627--6638, jun 2023.

\bibitem{Kwon_2014_Article}
J.-W. Kwon and D.~Chwa, ``Adaptive {B}idirectional {P}latoon {C}ontrol {U}sing a {C}oupled {S}liding {M}ode {C}ontrol {M}ethod,'' \emph{IEEE Trans. Intell. Transp. Syst.}, vol.~15, pp. 2040--2048, 2014.

\bibitem{Zhu_TITS_2022}
Y.~Zhu, Y.~Li, H.~Zhu, W.~Hua, G.~Huang, S.~Yu, S.~E. Li, and X.~Gao, ``Joint sliding-mode controller and observer for vehicle platoon subject to disturbance and acceleration failure of neighboring vehicles,'' \emph{IEEE Trans. Intell. Transp. Syst.}, vol. 2139, pp. 2345--2357, 2023.

\bibitem{Hu_TITS_2022}
X.~Hu, L.~Xie, L.~Xie, S.~Lu, W.~Xu, and H.~Su, ``Distributed model predictive control for vehicle platoon with mixed disturbances and model uncertainties,'' \emph{IEEE Trans. Intell. Transp. Syst.}, vol.~23, pp. 17\,354--17\,365, 2022.

\bibitem{Zhou_TITS_2022}
J.~Zhou, D.~Tian, Z.~Sheng, X.~Duan, G.~Qu, D.~Zhao, D.~Cao, and X.~Shen, ``Robust min-max model predictive vehicle platooning with causal disturbance feedback,'' \emph{IEEE Trans. Intell. Transp. Syst.}, vol.~23, pp. 15\,878--15\,897, 2022.

\bibitem{Zhai_TVT_2018}
C.~Zhai, F.~Luo, Y.~Liu, and Z.~Chen, ``Ecological cooperative look-ahead control for automated vehicles travelling on freeways with varying slopes,'' \emph{IEEE Transactions on Vehicular Technology}, vol.~68, no.~2, pp. 1208--1221, 2018.

\bibitem{Zhai_2022_Article}
C.~Zhai, F.~Luo, and Y.~Liu, ``Cooperative power split optimization for a group of intelligent electric vehicles travelling on a highway with varying slopes,'' \emph{IEEE Trans. Intell. Transp. Syst.}, vol.~23, pp. 4993--5005, 2022.

\bibitem{Na_2019_Article}
G.~Na, G.~Park, V.~Turri, K.~H. Johansson, H.~Shim, and Y.~Eun, ``Disturbance observer approach for fuel-efficient heavy-duty vehicle platooning,'' \emph{Vehicle System Dynamics}, vol.~58, no.~5, pp. 748--767, 2019.

\bibitem{Gao_2018_Article}
F.~Gao, X.~Hu, S.~E. Li, K.~Li, and Q.~Sun, ``Distributed {A}daptive {S}liding {M}ode {C}ontrol of {V}ehicular {P}latoon {W}ith {U}ncertain {I}nteraction {T}opology,'' \emph{IEEE Transactions on Industrial Electronics}, vol.~65, no.~8, pp. 6352--6361, Aug. 2018.

\bibitem{Lewis_2013_Book}
F.~L. Lewis, H.~Zhang, K.~Hengster-Movric, and A.~Das, \emph{Cooperative {C}ontrol of {M}ulti-{A}gent {S}ystems: {O}ptimal and {A}daptive {D}esign {A}pproaches}.\hskip 1em plus 0.5em minus 0.4em\relax Springer Science and Business Media, 2013.

\bibitem{Zheng_2021_Article}
Y.~Zheng, Y.~Bian, S.~Li, and S.~E. Li, ``Cooperative {C}ontrol of {H}eterogeneous {C}onnected {V}ehicles with {D}irected {A}cyclic {I}nteractions,'' \emph{{IEEE} Intelligent Transportation Systems Magazine}, vol.~13, no.~2, pp. 127--141, 2021.

\bibitem{Horn_2012_Book}
R.~A. Horn and C.~R. Johnson, \emph{Matrix {A}nalysis}.\hskip 1em plus 0.5em minus 0.4em\relax Cambridge University Press Cambridge, 2012.

\bibitem{Mesbahi_2010_Book}
M.~Mesbahi and M.~Egerstedt, \emph{Graph theoretic methods in multiagent networks}.\hskip 1em plus 0.5em minus 0.4em\relax Princeton University Press, 2010.

\bibitem{Shivakumar_1974_Article}
P.~N. Shivakumar and K.~H. Chew, ``A sufficient condition for nonvanishing of determinants,'' \emph{Proceedings of the American mathematical society}, pp. 63--66, 1974.

\bibitem{OlfatiSaber_2004_Article}
R.~Olfati-Saber and R.~M. Murray, ``Consensus {P}roblems in {N}etworks of {A}gents {W}ith {S}witching {T}opology and {T}ime-{D}elays,'' \emph{{IEEE} Transactions on Automatic Control}, vol.~49, no.~9, pp. 1520--1533, sep 2004.

\bibitem{Skogestad_2005_Book}
S.~Skogestad and I.~Postlethwaite, \emph{Multivariable feedback control: analysis and design}.\hskip 1em plus 0.5em minus 0.4em\relax John Wiley and Sons, 2005.

\bibitem{Dosovitskiy_2017_InProceedings}
\BIBentryALTinterwordspacing
A.~Dosovitskiy, G.~Ros, F.~Codevilla, A.~Lopez, and V.~Koltun, ``{CARLA}: {A}n open urban driving simulator,'' in \emph{Conference on robot learning}, 2017, pp. 1--16. [Online]. Available: \url{https://proceedings.mlr.press/v78/dosovitskiy17a.html}
\BIBentrySTDinterwordspacing

\bibitem{Hancock_2013_Article}
M.~W. Hancock and B.~Wright, ``A policy on geometric design of highways and streets,'' \emph{American Association of State Highway and Transportation Officials: Washington, DC, USA}, vol.~3, 2013.

\end{thebibliography}
\bibliographystyle{IEEEtran}

\vspace{-1cm}

\begin{IEEEbiography}[{\includegraphics[width=2.3cm,clip,keepaspectratio]{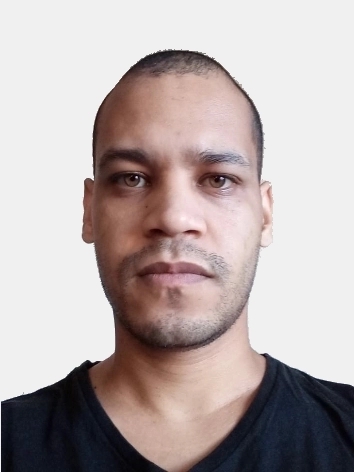}}]{Emerson Alves da Silva} 
received his B.S. in Systems Engineering in 2021 and M.S. in Electrical Engineering in 2022, both from the Universidade Federal de Minas Gerais (UFMG), He is currently a Ph.D. student of Electrical Engineering in UFMG, Brazil, in the research field of control and modeling of autonomous and connected vehicles.
\end{IEEEbiography}
\vspace{-1cm}

\begin{IEEEbiography}[{\includegraphics[width=2.3cm,clip,keepaspectratio]{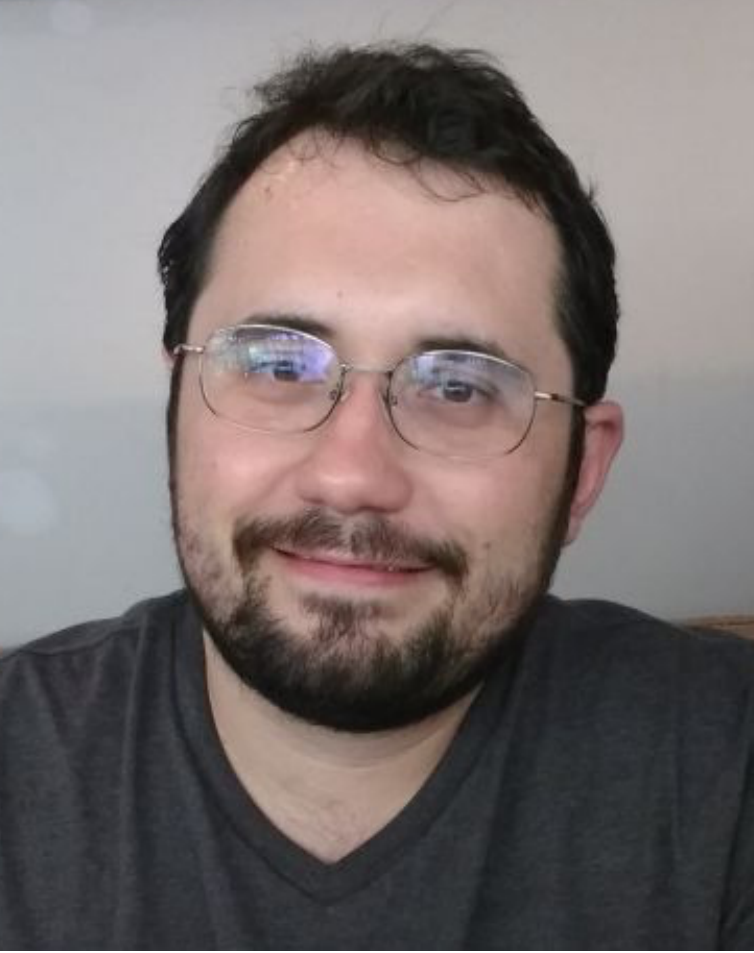}}] 
	{Leonardo Amaral Mozelli} received the B.Sc., M.Sc., and Ph.D. degrees (Electrical Engineering) from Universidade Federal de Minas Gerais (UFMG), in 2006, 2008, and 2011, respectively. He is an Associate Professor at Department of Electronic Engineering, UFMG, since 2017, teaching in the interlinked areas of signal processing, dynamic systems, and automatic control.  Research topics: multi-agent systems, adaptive signal processing, robust control theory, and technologies for sustainable development.
\end{IEEEbiography}
\vspace{-1cm}

\begin{IEEEbiography}[{\includegraphics[width=2.3cm,clip,keepaspectratio]{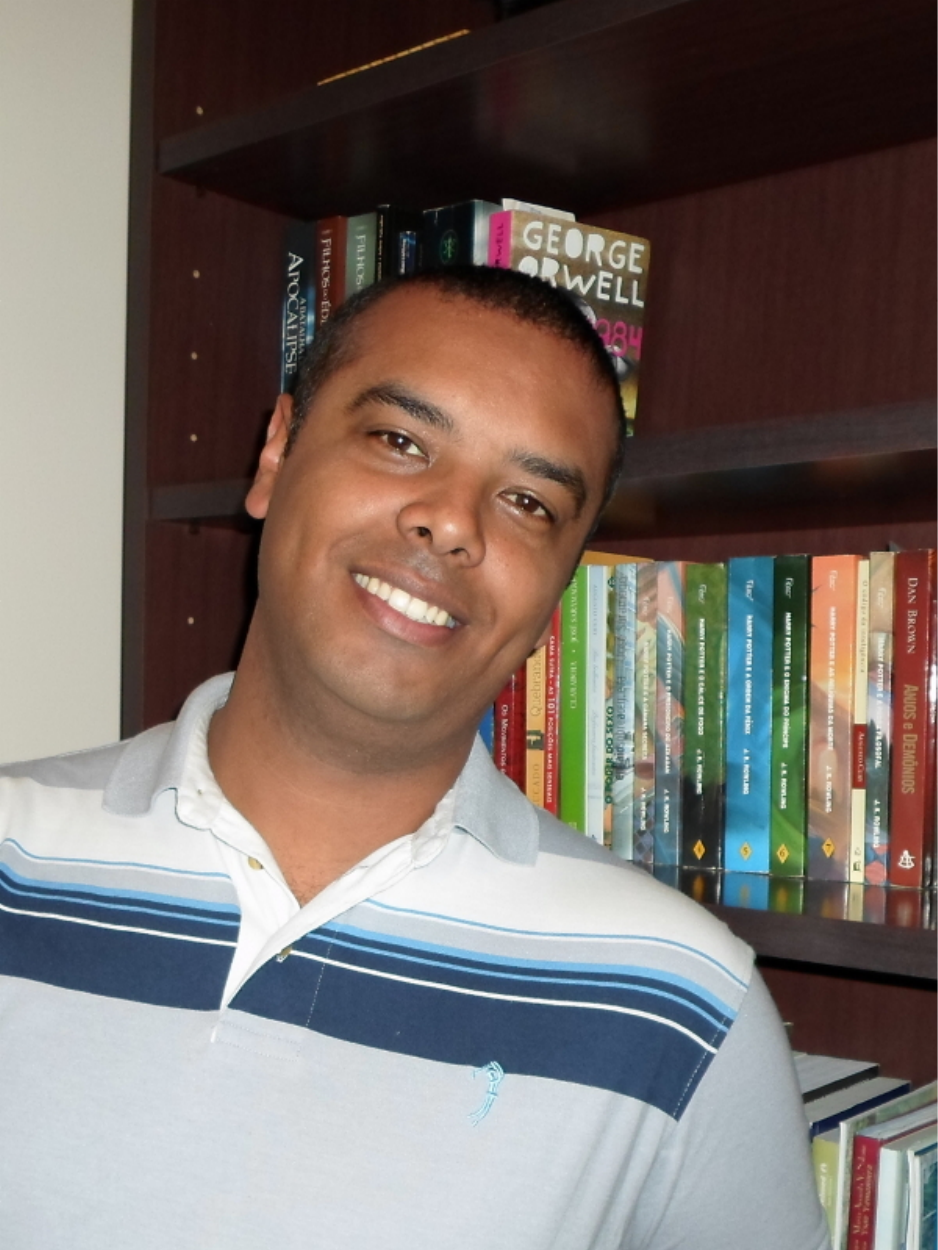}}]{Armando Alves Neto}
received the B.S.E. degree in Automation and Control Engineering from the Universidade Federal de Minas Gerais in 2006, and S.M. and Ph.D. degrees in Computer Science from UFMG in 2008 and 2012, respectively. He is an Assistant Professor in the Department of Electronic Engineering at UFMG. Research interests include real-time motion planning, multi-agent control, robust control, and collision avoidance strategies.
\end{IEEEbiography}

\vspace{-1cm}

\begin{IEEEbiography}[{\includegraphics[width=2.3cm,clip,keepaspectratio]{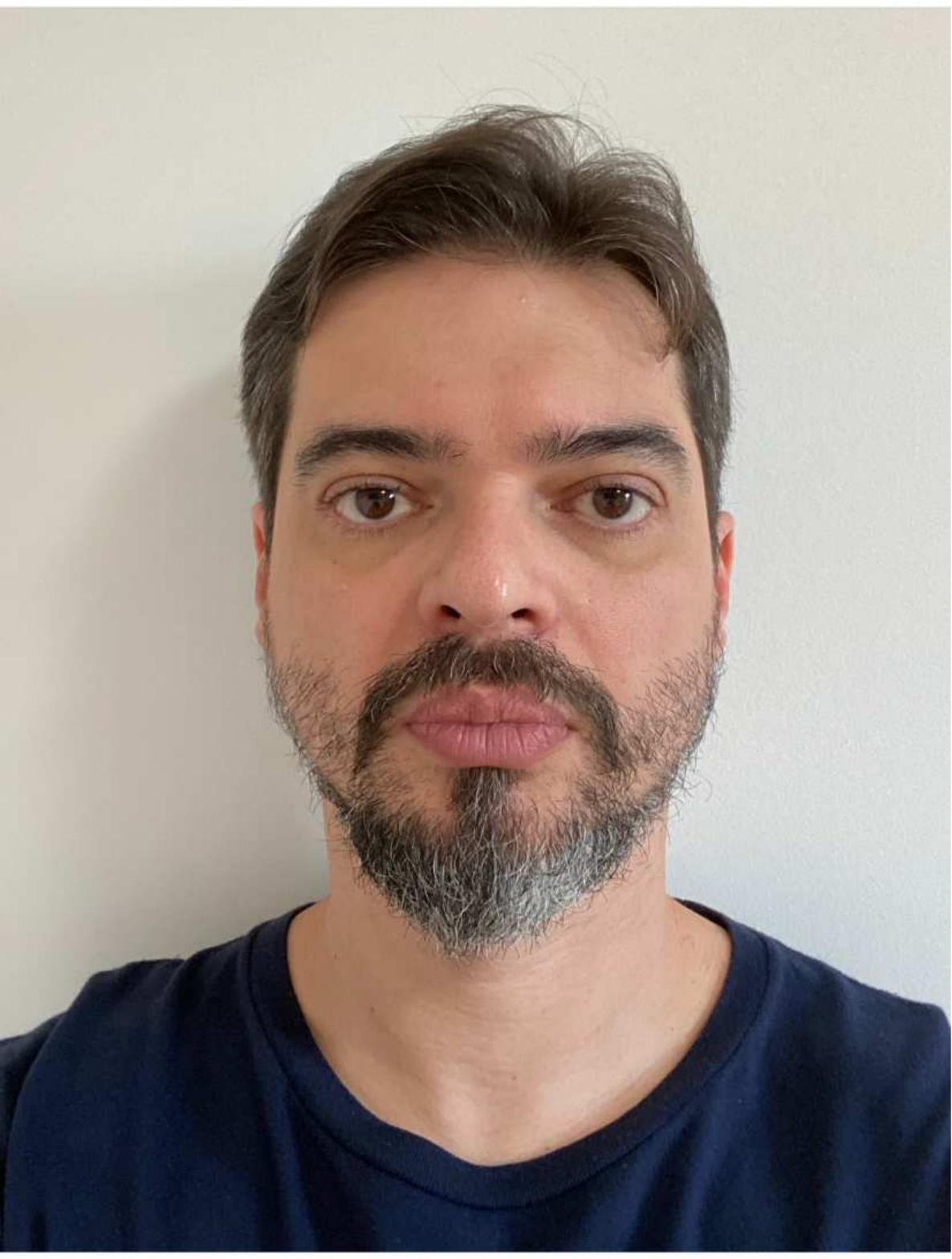}}]{Fernando de Oliveira Souza}
obtained his B.S. (2003) degree in Control and Automation Engineering from the Pontif\'icia Universidade Cat\'olica de Minas Gerais, Brazil, and his M.S. (2005) and Ph.D. (2008) degrees in Electrical Engineering from the Universidade Federal de Minas Gerais (UFMG), Brazil. He is currently an associate professor at the Department of Electronic Engineering, UFMG. His research interests include multi-agent systems, time-delay systems, and robust control.
\end{IEEEbiography}

\vfill
\end{document}